\documentclass{llncs}
\usepackage{paralist}
\usepackage{etex}
\usepackage{amsmath}
\usepackage{amssymb}
\usepackage{wrapfig,booktabs}
\usepackage{pdfsync}
\usepackage[all]{xy}
\usepackage{color,xspace}
\usepackage{hyperref}
\usepackage[ruled,linesnumbered]{algorithm2e}
\usepackage{environ}
\usepackage{tikz}
\usepackage{pgfplots}
\usepackage{qtree}
\usepackage{caption}
\usepackage[final]{commenting}
\usetikzlibrary{matrix}

\newcommand{\keys}{S}
\newcommand{\keysval}{s}
\newcommand{\queries}{\Sigma}

\newcommand{\queriesval}{\sigma}
\newcommand{\obs}{O}
\newcommand{\obsval}{o}
\newcommand{\view}{\mathit{view}}
\newcommand{\up}{\mathit{upd}}

\newcommand{\partition}{R}
\newcommand{\block}{r}
\newcommand{\randsecret}{X}
\newcommand{\randguess}{\hat{\randsecret}}
\newcommand{\randkey}{Y}
\newcommand{\randblock}{Z}
\newcommand{\memblock}{\mathcal{B}}
\newcommand{\cache}{C}
\newcommand{\ages}{\mathcal{A}}
\newcommand{\assoc}{A}

\newcommand{\maxleak}{r_{\max}}
\newcommand{\hit}{\mathsf{H}}
\newcommand{\miss}{\mathsf{M}}

\newcommand{\FIFO}{\text{FIFO}}
\newcommand{\PLRU}{\text{PLRU}}
\newcommand{\LRU}{\text{LRU}}
\newcommand{\att}{\text{att}}

\newcommand{\footprint}{\mathit{fp}}
\newcommand{\sizeof}[1]{\left| #1 \right|}
\newcommand\range[1]{\mathop{ran}(#1)}
\newcommand\ABS{\mathit{Abs}}
\newcommand\emptyc{e}
\newcommand\filledc{f}

\declareauthor{bk}{Boris}{red!50}
\declareauthor{pc}{Pablo}{blue}
\declareauthor{jr}{Jan}{orange}

\makeatletter
\newcommand{\customlabel}[2]{%
   \protected@write \@auxout {}{\string \newlabel {#1}{{#2}{\thepage}{#2}{#1}{}} }%
   \hypertarget{#1}
}
\makeatother

\pgfplotsset{width=7cm,compat=newest,xlabel style={font=\small}}

\newcommand{\AddCurve}[4]{% ensure that the filename has a suffix, e.g. .txt
       \addplot[mark=#3, color=#2,style=#4] plot file {#1};
}

\definecolor{color1}{RGB}{189, 0, 38}           %red
\definecolor{color2}{RGB}{116, 196, 118}        %green
\definecolor{color3}{RGB}{252, 146, 114}        %red

\newif\ifshowproof
\showprooftrue
\NewEnviron{Proof}{%
    \ifshowproof%
        \begin{proof}%
            \BODY\qed
        \end{proof}%
    \fi%
}%

\begin{document}
\pagestyle{headings} % Adds page numbers.

\title{Security Analysis of Cache Replacement Policies}

\author{Pablo Ca{\~n}ones\inst{1} \and Boris K{\"o}pf\inst{1} \and
    Jan Reineke\inst{2}}

\institute{
IMDEA Software Institute, Madrid, Spain \\
                \email{\{pablo.canones,boris.koepf\}@imdea.org}\\[2ex]
\and
Saarland University, Saarland Informatics Campus, Saarbr{\"u}cken, Germany\\
\email{reineke@cs.uni-saarland.de}
}
    
\maketitle
\begin{abstract}
  Modern computer architectures share physical resources between
  different programs in order to increase area-, energy-, and
  cost-efficiency.  Unfortunately, sharing often gives rise to side
  channels that can be exploited for extracting or transmitting
  sensitive information. We currently lack techniques for systematic
  reasoning about this interplay between security and efficiency. In
  particular, there is no established way for quantifying security
  properties of shared caches.

  In this paper, we propose a novel model that enables us to
  characterize important security properties of caches. Our model
  encompasses two aspects: (1) The amount of information that can be
  {\em absorbed} by a cache, and (2) the amount of information that
  can effectively be {\em extracted} from the cache by an adversary.
  We use our model to compute both quantities for common cache
  replacement policies (FIFO, LRU, and PLRU) and to compare their
  isolation properties.  We further show how our model for information
  extraction leads to an algorithm that can be used to improve the
  bounds delivered by the CacheAudit static analyzer.
\end{abstract}

\section{Introduction}
Modern computer architectures share physical resources across
different programs in order to increase area-, energy-, and
cost-efficiency. Examples of commonly shared resources are caches,
branch prediction units, DRAM, and disks.

Unfortunately, sharing poses a threat to security: even if programs
are completely isolated on a logical level, sharing a physical
resource usually means that one program's resource usage pattern can
be observed by the other. This constitutes a channel that can be
exploited for extracting or transmitting sensitive information. While
this kind of vulnerability has been known for
decades~\cite{lampson1973note}, its severity has become painfully
apparent with a stream of highly effective side-channel attacks.  One
shared resource that has been the objective of a large number of
attacks are CPU caches,
e.g.~\cite{osvikshamir06cache,Bernstein05cache-timingattacks,Aciicmez07Simp,Aciicmez07Diff,GullaschBK11,YaromF14,LiuYGHL15}.

From a security point of view it would be ideal to completely
eliminate side channels through the cache by design, as
in~\cite{TiwariOLVLHKCS11,zwsm15}, or to flush the cache between
accesses of two different parties. Unfortunately, such conservative
approaches also partially void the performance benefits of sharing. In
many practical scenarios, designers will opt for less conservative
solutions that offer ``sufficient'' degrees of security together with
high performance. However, while there is a large body of work on
evaluating the impact of different cache designs on performance, there
are no established metrics for evaluating their security, which
prevents principled decision-making in that design space.

\subsubsection{Approach.}
In this paper, we address this problem by introducing a novel approach to quantify the security of caches, in particular: their replacement
policies. Our approach aims to answer the following questions, which
capture two natural aspects of isolation between programs that share
the cache:

\begin{asparaenum}
\item[{\bf Q1}] {\em How much information about a computation is absorbed by the cache?}\\ There are two challenges
  involved with this question. The first is to identify a meaningful
  measure for the information contained in a given cache state. 
  The second is to characterize the set of possible computations, which may induce different cache states.
  To make assertions about the security of the
  cache architecture (rather than about the security of a specific
  program running on top of a cache architecture) such a characterization needs to encompass a sufficiently
  general class of programs.

\item[{\bf Q2}] {\em How much information can an adversary extract
    from the cache state?}\\ The challenge for answering this question
  is that an adversary can only learn about the cache state by
  probing, that is, by performing memory accesses and measuring their
  latency. However, probing also modifies the cache state and thus can
  reduce its information content. With the exception of one approach
  that encompasses secrets that change over time~\cite{mardziel14},
  existing models of quantitative information flow do not account for
  this scenario because they either consider only single
  probes~\cite{smith09} or assume the secret remains unchanged by the
  probing~\cite{koepfbasin07,BorealeP12,alvim2012measuring}.
\label{q2}
\end{asparaenum}
\begin{asparaenum}
\item[{\bf A1}] For answering Q1, we characterize the absorbed
  information as the number of reachable cache states, which
  essentially captures the information that programs leak {\em into}
  the cache. For a single program, this amount can be bounded
  using existing static analysis
  tools~\cite{doychev2015cacheaudit}. For abstracting from a specific
  program, we draw inspiration from the working set
  model~\cite{DenningWSM} and characterize programs in terms of their
  footprint, i.e., the number of memory blocks they use. We then show
  how (and under which assumptions) the footprint alone can be used to
  characterize the absorption of a given replacement policy, leading
  to a program-independent measure.
\item[{\bf A2}] For answering Q2, we put forward a novel model to
  quantify the ``extractable'' information about the cache state. We
  consider an adversary that adaptively provides inputs and observes
  the outputs. The key difference to existing models of adaptive
  attacks~\cite{koepfbasin07,BorealeP12} is that our model is based on
  a Mealy machine in which each input triggers a state transition,
  which may erase information about its origin. As in existing models,
  we first characterize the revealed information in terms of a
  partition of the set of secrets (here: initial states of the
  machine). We then evaluate this partition with established measures
  of leakage to quantify the corresponding amount of information. By
  considering the maximum leakage w.r.t. all possible inputs to the
  Mealy machine, we obtain an upper bound on the information that any
  adaptive adversary can extract.  We present an algorithm that
  computes such bounds for given Mealy machines.
\end{asparaenum}

\subsubsection{Results.}
We put our models and algorithms to work for the quantification of
absorption and extraction properties of common cache replacement
policies, namely FIFO, LRU, and PLRU. We highlight the following
results; see the paper for more details.
\begin{compactitem}
\item We show that the relative security ranking of cache replacement
  policies varies widely depending on the memory demand of the
  program. For example, FIFO can provide the best security when memory
  demand is low, whereas LRU generally provides the best security. Our
  results show that PLRU generally offers worse security than the other
  replacement policies.
\item We show that our algorithm for information extraction can be
  used for improving the cache-state counting of the CacheAudit static
  analyzer~\cite{doychev2015cacheaudit}. Our experimental results show
  that this significantly improves the bounds delivered by CacheAudit,
  leading to gains of up to 50 bits for AES 256.
\end{compactitem}

\subsubsection{Contribution.}
In summary, our conceptual contribution is to propose novel measures
for quantifying isolation properties of shared caches.  Our practical
contribution is to perform the first security analysis of common cache
replacement policies.

\section{The Model}

\subsection{Caches as Mealy Machines}
\label{sec:replacement}
Caches are fast but small memories that store a subset of the main memory's contents to bridge the latency gap between the CPU and the main memory.  To profit from spatial locality and to reduce management overhead, main memory is logically partitioned into a set $\memblock$ of memory blocks.  Each block is cached as a whole in a cache line of the same size.  When accessing a memory block, the cache logic has to determine whether the block is stored in the cache (``cache hit'') or not (``cache miss'').

In this paper, we model caches as Mealy machines, that is, finite automata that map sequences of accessed memory blocks to sequences of hits and misses. We begin by recalling the definition of a Mealy machine before we specialize it to the case of caches.
\begin{definition}
  A (deterministic) {\em Mealy machine} $M$ is a five-tuple consisting of 
\begin{compactitem}
\item $\keys$: a finite set of {\em states}, 
\item $\queries$: a finite set of {\em inputs}, 
\item $\obs$: a finite set of {\em outputs} (or {\em observations}), 
\item $\up\colon\keys\times\queries\rightarrow \keys$: a {\em transition function}, and 
\item $\view\colon\keys\times\queries \rightarrow \obs$: an {\em observation function}
\end{compactitem}
\end{definition}
For casting
caches as Mealy machines, we use memory blocks as inputs,
i.e. $\queries=\memblock$, and cache hits ($\hit$) and misses
($\miss$) as observations, i.e., $\obs=\{\hit,\miss\}$.  For defining
the set of states $\keys$, recall that caches are commonly partitioned
into independent equally-sized \emph{cache sets} whose size $\assoc$
is called the {\em associativity} of the cache. For each block there
is a single cache set that stores it.

For simplicity of presentation we focus on caches with a single set. Since cache sets behave independently from each other, the technique is generalizable to several sets by focusing each time on the blocks stored in a particular set. We model a cache set as a function that assigns an age in $\ages := \{ 0, \ldots , \assoc-1, \assoc \}$ to each memory block.
$$\keys= \{ c \in \memblock \rightarrow \ages \mid~ 
     \forall b_1,b_2 \in \memblock: b_1 \neq b_2~\Rightarrow \hfill c(b_1) \neq c(b_2) \vee c(b_1) = c(b_2) = \assoc)\}\ .
$$
Here, the youngest block has age $0$ and the oldest cached block has age $\assoc-1$. Age $\assoc$ means that a block is not cached; it is the only age that can be shared by multiple blocks.

With this, the observation function $\view_b=\view(\cdot,b)$ is naturally defined as
\begin{equation*}
\view_b(c) = \begin{cases}
  \hit   & \text{if } c(b) < \assoc\\
  \miss &  \mbox{else}
\end{cases}
\end{equation*}
The transition function $\up_b=\up(\cdot,b)$ is specified by:
\begin{equation}
\label{eq:cache_update}
 \up_b(c)(b') =
    \begin{cases}
      c(b') & \text{if }  b' \neq b \wedge c(b') = \assoc\\
      0 & \text{if } b' = b\wedge 
c(b)=\assoc \\
      c(b')+ 1 & \text{if } b' \neq b \wedge c(b')<\assoc  \wedge 
c(b)=\assoc \\
      \Pi_{c(b)}(c(b')) & \text{if } c(b') < \assoc \wedge c(b)<\assoc \\
    \end{cases}
\end{equation}

This transition function models \emph{permutation} replacement policies as defined in~\cite{abel2013measurement}. Upon a miss, $c(b)=A$, the accessed block is placed at the beginning of the cache, increasing the ages of younger blocks and evicting the block with age $\assoc-1$. In the case of a hit, each replacement policy reorders the blocks in a certain way, determined by the \emph{permutation} function $\Pi_\alpha(\alpha'):\ages\rightarrow \ages$; it modifies the current age $\alpha'$ of a block according to a \emph{base} age $\alpha$.

\ifshowproof\else Each replacement policy has its own permutation
function: FIFO does not reorder the blocks, LRU sets the age of the accessed
block to 0, and PLRU behaves similar to LRU but
with a more complex reorganization. We refer to the Mealy machines corresponding to LRU, PLRU, and FIFO caches by $M_\LRU, M_\PLRU$, and $M_\FIFO$, respectively.

The formalization of these policies, as well as the proofs of all
technical results are contained in the extended version of this paper.
\fi

\ifshowproof%
We introduce the definition the three permutation functions following the model developed in \cite{abel2013measurement}.

The FIFO (First In First Out) replacement policy does not change the ages of the blocks upon cache 
hits. It is is thus modeled as the identity permutation.
\begin{equation}
\label{eq:fifo_permutation}
\Pi^{\mathit{FIFO}}_\alpha(\alpha') = \alpha'
\end{equation}

The LRU (Least Recently Used) replacement policy sets the age of an accessed block to $0$ upon a 
cache hit, making sure that the least-recently-used blocks get evicted upon misses. 
Formally, we cast this behavior as
\begin{equation}
\label{eq:lru_permutation}
\Pi^{\mathit{LRU}}_\alpha(\alpha') = 
\begin{cases}
  0 & \text{if } \alpha' = \alpha\\
  \alpha' + 1 & \text{if } \alpha' < \alpha\\
  \alpha' & \text{if } \alpha' > \alpha
\end{cases}
\end{equation}

The PLRU (Pseudo Least Recently Used) is similar to LRU, but with a more complex permutation function. For an associativity which is a power of two, PLRU represents each cache set as a full binary tree 
storing the blocks at its leaves, and each non-leaf stores a bit which 
represents an arrow pointing to one of the children. Upon a cache miss, the 
block to 
be evicted is determined by following the arrows starting from the root. Upon 
any cache access (regardless whether it is a hit or a miss), the arrows on the way to the 
accessed block are 
flipped. Figure~\ref{fig:plru-example} shows an example of two consecutive 
cache hits in a 4-way cache.
In this paper we assume that the associativity is always a power of two.

\begin{figure}[tb]
\centering
\includegraphics[width=0.6\textwidth]{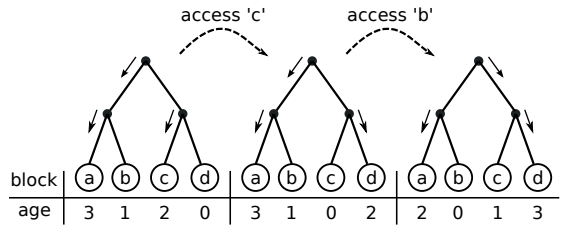}
\caption{An example of two consecutive cache hits with 
PLRU.}\label{fig:plru-example}
\end{figure}

We formally define this PLRU permutation policy $\Pi^{\mathit{PLRU}}$ 
as

\begin{equation}
\label{eq:plru_permutation}
\Pi^{\mathit{PLRU}}_\alpha(\alpha') =
\begin{cases}
 0 &\text{if } \alpha' = \alpha \\
 \alpha' &\text{if } \alpha \text{ even} \wedge \alpha' \text{ odd}\\
 \alpha' + 1 &\text{if } \alpha \text{ odd} \wedge \alpha' \text{ even}\\
 2\cdot\Pi^{\mathit{PLRU}}_{\lfloor \alpha/2\rfloor}(\lfloor \alpha'/2\rfloor) &  \text{ 
otherwise}
\end{cases}
\end{equation}
\fi

\subsection{Quantifying Absorption and Extraction}\label{sec:Markov_chains} 
 We characterize absorption and extraction in terms of the interactions of two agents, a victim and an adversary.
\begin{asparaitem}
\item The {\em victim} first chooses a secret, such as a cryptographic key. We model this using a random variable~$\randsecret$. The victim then uses this secret as input to a program that he runs to completion (or preemption) on a platform with a cache. We capture the effect of the victim's computation on the cache state in terms of a finite sequence of blocks from the set of \emph{victim's blocks} $B_v$, where $B_v\subseteq\memblock$. The cache uses this sequence as inputs to transition from an initial state 
  to the \emph{victim's state}. We model the victim's state using
  a random variable $\randkey_v$ that takes values in a set $\keys_v\subseteq \keys$, i.e. $\range{\randkey_v}=\keys_v$.

\item The {\em adversary} then runs a program on the same platform, which enables him to make observations about the state of the cache by measuring the latency of its memory accesses.\footnote{In the literature, this is known as an {\em access-based adversary}, e.g.~\cite{NeveS06}.} We model the adversary's actions in terms of a finite sequence of blocks from the subset of \emph{attacker's blocks} $B_a\subseteq\memblock$. Using the sequence of blocks as inputs, the cache transitions from the victim's state returning a sequence of hits and misses that we model with the random variable~$\randblock_a$, $\range{\randblock_a}\subseteq \obs^*$. We make the random variable dependent on the attacker since he can choose the sequence of blocks. Based on these observations, the adversary tries to guess the secret. We model the guess in terms of the random variable~$\randguess$.\footnote{Note that, while $\randkey_v$ and $\randblock_a$ are given in terms of inputs and outputs of the Mealy machine representing the cache, we do not assume any particular structure on $\randsecret$ and $\randguess$.}
\end{asparaitem}
We say that an attack is successful if the adversary correctly guesses the secret, i.e. if $\randsecret=\randguess$. 

We now give a high-level operational motivation for our definitions of information absorption and extraction, in terms of a bound on the probability of a successful attack. We assume that the distribution of each of these random variables depends only on the outcome of the previous one, i.e., that the distribution of cache states depends only on the secret, and that the adversary's observations depend only on the state of the cache. Then we can cast the dependencies between these random variables in terms of the following Markov chain:
\begin{equation}\label{eq:markov} 
\underset{\text{Secret}}{\randsecret} \overset{\overset{\text{Victim}}{|}}{\longrightarrow} \underset{\text{Cache State}}{\randkey_v} \overset{\overset{\text{Adversary probe}}{|}}{\longrightarrow} \underset{\text{Observation}}{\randblock_a} \overset{\overset{\text{Adversary guess }}{|}}{\longrightarrow} \underset{\text{Guess}}{\randguess} 
\end{equation}

The following result bounds the probability of a successful attack, i.e. $P(\randsecret=\hat{\randsecret})$, in terms of the size of the ranges of $\randkey_v$ and $\randblock_a$, respectively.

\begin{theorem}\label{thm:markov} 
  \begin{align} \label{eq-ex} P(\randsecret=\hat{\randsecret})&\leq
    \max_{x\in ran(\randsecret)} P(\randsecret=x)\cdot
    \sizeof{\range{\randblock_a}}\\
 \label{eq-abs}
    P(\randsecret=\hat{\randsecret})&\leq \max_{x\in ran(\randsecret)}
    P(\randsecret=x)\cdot \sizeof{\range{\randkey_v}}
\end{align} 
\end{theorem} 
\begin{Proof} 
The result is consequence of the fact that the reduction of min-entropy of $\randsecret$ when observing a jointly distributed random variable is upper-bounded by the size of the range of that variable~\cite{smith09}. We cast this result in terms of Markov chain notation as in~\cite{doychev2015cacheaudit}, and apply it to both $\randkey_v$ and $\randblock_a$.  
\end{Proof}
For an attacker that follows a deterministic strategy, the value of
$\randblock_a$ is determined by the value of $\randkey_v$. Therefore
$\sizeof{\range{\randblock_a}}\le\sizeof{\range{\randkey_v}}$, which
implies that~\eqref{eq-ex} leads to better security guarantees
than~\eqref{eq-abs}.

 Whenever additionally the value of $\randkey_v$ is determined by that
 of $\randsecret$ and $\randsecret$ is uniformly distributed, the
 bounds given by Theorem~\ref{thm:markov} are tight, in the sense that they can be achieved by computationally unbounded adversaries.

In this paper, we will use $\sizeof{\range{\randkey_v}}$ to capture the amount of information that is {\em absorbed} by the cache, and we will use $\sizeof{\range{\randblock_a}}$ to capture the amount of information that the adversary can {\em extract} from the cache. The operational significance of these quantities follows from Theorem~\ref{thm:markov}. We discuss how these quantities can be computed in Section~\ref{sec:absorption} and \ref{sec:extraction}, respectively.

\section{Absorption of Information}\label{sec:absorption}

\begin{wrapfigure}{r}{0.3\textwidth}
\vspace{-1ex}
\centering 
\includegraphics[width=0.25\textwidth]{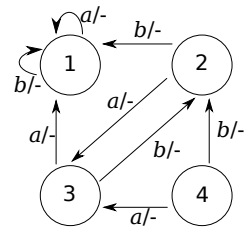} 
\caption{Example of Mealy machine} \label{fig:example_s0} 
\vspace*{-5ex}
\end{wrapfigure}

In this section we characterize the information absorption of different
cache replacement policies. That is, we characterize
$\range{\randkey_v}$ from \eqref{eq:markov} as a subset
$\keys_v\subseteq\keys$ of reachable victim's states of the Mealy machine
representing the cache.

Before we give the formal definition we note that the absorbed information depends on two things: the initial state of the Mealy machine and the inputs of the victim. To see the effect of the initial state $\keysval_0\in\keys$, consider the Mealy machine in Figure~\ref{fig:example_s0} and assume that the victim may use any sequence of inputs from $\queries_v^*=\{a,b\}^*$. If we start from the state $\keysval_0=1$ only that one state is reachable, $\keys_v=\{1\}$; if $\keysval_0=2,3$ then $\keys_v=\{1,2,3\}$ and finally if $\keysval_0=4$ then $\keys_v=\keys$. 

We capture the victim's inputs as a trace $t\in\queries_v^*$. This leads to the following definition of $\sizeof{\range{\randkey_v}}$.  
\begin{definition} \label{def:absorbed_information}
  We define the {\em absorbed information} of a Mealy machine
  $M=(\keys,\queries,\obs,\up,\view)$ w.r.t an initial state $\keysval_0$ and a set of traces $T\subseteq\queries_v^*$ as
\begin{equation*}
\ABS(M,\keysval_0,T)=\sizeof{\{\keysval\in \keys\mid \exists t\in T\colon \up_t(\keysval_0)=s\}},
\end{equation*}
\end{definition}
In the above definition of absorption, the set of traces $T$ is a parameter.  For a given program, existing static analysis techniques can be used to compute approximations of the set of traces $T$ and the induced absorption of a particular cache, modeled by a Mealy machine $M$.  In Section~\ref{sec:experiments} we present the results of a static analysis of two AES implementations.

In this section, our goal is to characterize the absorption properties of caches \emph{independently} of a particular program.  A worst case approach to this end is to study absorption under all possible traces $T=B_v^*$, given a set of memory blocks $B_v$.  For this, we first state several general results in Section~\ref{sec:caches_intake}, which show that the absorption of caches is independent of the particular set of memory blocks $B_v$ being accessed, and only depends on its size, $|B_v|$.  In Section~\ref{sec:Cache_abs}, we then use these general results to derive concrete results on the absorption properties of caches under LRU, FIFO, and PLRU replacement.

\subsection{Data Independence of Permutation Replacement Policies}\label{sec:caches_intake}
\paragraph{Initial State}
Absorption, as defined in Definition~\ref{def:absorbed_information} depends on the initial state of the Mealy machine.
Considering programs that may access the set of memory blocks $B \subseteq \memblock$, two types of initial states for caches are particularly interesting:

\begin{definition}
We say that a cache state $c\colon \memblock\rightarrow\ages$ is
\begin{compactenum}
\item {\em empty w.r.t. }$B$ if $c(B) = \{c(b) \mid b \in B\} =\{\assoc\}$. That is, none of
  the blocks in $B$ are cached.
\item {\em filled with} $B$ if
  $c(B)=\{0,\dots,\min(\assoc,\sizeof{B}-1)\}$. That is, the blocks in $B$
  occupy the cache. If $B$ contains less blocks than cache lines, we
  require that the first $|B|$ lines are filled.
\end{compactenum}
\end{definition}
The notions of \emph{empty} and \emph{filled} cache states are relative to a set of memory blocks.
We will consider empty and filled cache states relative to the memory blocks accessed by the victim, $ B_v$.
To conservatively capture the power of an attacker, ages without a victim's block mapped to them will be assumed to hold the attacker's memory blocks not accessible for the victim, that is, blocks from the set $B_a\setminus B_v$.

\paragraph{Data Independence}
The following result is central for our program-independent analysis of cache
replacement policies. It shows that absorption can be characterized
independently of the particular set of blocks $B$ that the victim may access:
\begin{theorem}\label{thm:dataindependence}
Whenever $\sizeof{B_1}=\sizeof{B_2}$, and $c_1$ is empty (filled) w.r.t. $B_1$ and $c_2$ empty (filled) w.r.t. $B_2$, then 
$$\ABS(M,c_1,B_1^*)=\ABS(M,c_2,B_2^*).$$
\end{theorem}
The proof of Theorem~\ref{thm:dataindependence} follows from the
following lemma and the observation that one can define
bijections between all sets of equal cardinality.
\begin{lemma}\label{lem:dataindep}
Let $f\colon \memblock\rightarrow\memblock$ be a bijection. Then
\begin{equation*}
\ABS(M,c_0,B^*)=\ABS(M,c_0\circ f^{-1},(f(B))^*).
\end{equation*}
\end{lemma}
\begin{Proof}
For any state $c$ and any blocks $b,b'$ we have $\up_{b}(c)(b')=\up_{f(b)}(c\circ f^{-1})(f(b'))$. This is because the transition functions of caches do not consider the block itself, they only perform equality checks or compare the ages. Since $f$ is a bijection and $c\circ f^{-1}(f(b))=c(b)$ for all $b$, the output of the update functions coincides. Therefore, every update from $c_0$ with any trace of blocks $b_1\dots b_n$ produces the same state as updating $c_0\circ f^{-1}$ with the trace $f(b_1) \dots f(b_n)$. 
\end{Proof}

We focus on filled and empty initial states since they represent the two extremes for the information absorption. Consider a \emph{partially filled} state $c$, that is, where there is a sequence of distinct blocks $b_0\ldots b_n$ with $n\leq\min(\assoc,\sizeof{B}-1)$ such that $c(b_i)=i$ for $i\leq n$. Then, any state reachable from $c$ by inputting a trace $t\in B^*$ is reachable from an empty one $c_e$ with the trace $t'=b_n\ldots b_0 t$. Since $c_e$ is empty, we load the blocks $b_0\ldots b_n$ in reverse order; these access produce misses and so, after the updates, $\up_{b_0}\cdots\up_{b_n} (c_e)(b_i)=i$, see \eqref{eq:cache_update}. Therefore $\ABS(M,c,B^*)\leq \ABS(M,c_e,B^*)$. Using this argument we can see that, for the same set of memory blocks, the value of the absorbed information is the smallest when starting on a filled state and is the largest when starting on an empty state.

An important consequence of Theorem~\ref{thm:dataindependence} is
that, given an identical status, i.e. empty or filled, of the initial
state, the amount of absorbed information depends only on the number
of blocks in $ B_v$. We call this number the {\em footprint}
and denote it by $\footprint=\sizeof{ B_v}$.  This terminology
is loosely connected with the notion of a memory footprint as used in
the theory of locality~\cite{xiang13hotl}. Theory of locality defines 
the footprint as the number of distinct memory blocks accessed during a time window, i.e.  on a trace of a given length.
In our case we consider this length to be unbounded so the trace is the whole execution of the program. This motivates the specialization of the
definition of the absorbed information in terms of the footprint, namely
\begin{equation*} 
\ABS_x(M,\footprint)=\ABS(M,c_0,( B_v)^*)\ ,
\end{equation*}
where we use the subscript $x=\emptyc$ to denote that $c_0$ is
empty w.r.t. $ B_v$, and $x=\filledc$ to denote that $c_0$ is filled
w.r.t. $ B_v$.

\subsection{Analysis of Cache Replacement Policies}
\label{sec:Cache_abs}
Next we give a summary of our program-independent analysis of the absorption for
each replacement policy. 
\subsubsection{Results for Filled Caches}
For some replacement policies, when the cache is filled and the footprint is small enough, some cache states are unreachable from the initial state, which reduces the information absorption. The details for each policy are given below. 
In case \emph{every} state of the cache is reachable, we count all the possible feasible mappings of $\footprint$ blocks to the set of ages $\ages$. Then the absorbed information is the number of \emph{k-permutations of n} of the memory blocks, i.e., the number of different ordered arrangements of $\footprint$ blocks in a sequence of up to $\assoc$ elements.

\begin{proposition}\label{prop:first}
For $M_\LRU$, the absorbed information for a filled cache is:
$$
\ABS_\filledc(M_\LRU,\footprint)=\begin{cases}
\footprint!&\text{if }\footprint<\assoc,\\
\frac{\footprint!}{(\footprint-\assoc)!}&\text{if }\footprint\geq \assoc.
\end{cases}
$$
\end{proposition}

\begin{Proof}
Every cache state is reachable, even if starting from a filled cache. Since upon a hit the input block obtains age zero, any state can be reached simply accessing the target state's blocks from oldest to youngest, see \eqref{eq:lru_permutation}.  Therefore, the absorbed information is the number of k-permutations of n of the available blocks. If $\footprint<\assoc$, all the victim's blocks fit in the cache and the k-permutations of n are $\footprint!$. If $\footprint\geq\assoc$, we can only fit $\assoc$ blocks in the cache, which gives the value $\footprint!/(\footprint-\assoc)!$. 
\end{Proof}

\begin{proposition}
For $M_\FIFO$, the absorbed information for a filled cache is:
$$
\ABS_\filledc(M_\FIFO,\footprint) =
\begin{cases}
1 & \text{ if }\footprint\leq\assoc, \\
\assoc+1 & \text{ if }\footprint=\assoc+1,\\
\frac{\footprint!}{(\footprint-\assoc)!} & \text{ if }\footprint>\assoc+1.
\end{cases}
$$
\end{proposition}
\begin{Proof}
An important property of FIFO is that it does not reorder cached blocks upon hits, see \eqref{eq:fifo_permutation}. So, for $\footprint\leq\assoc$ and a filled cache state, all the accesses are hits and leave the state in its original form. 

When $\footprint=\assoc+1$ every reachable state contains all but one of the $\footprint$ many blocks. For a set of blocks $\{b_1,\ldots, b_{\assoc+1}\}$ assume that the cache is initially in the following state:
$$[b_1,b_2, \ldots,b_\assoc],$$
where the leftmost element of the list has age zero and the one on the right is the oldest. An access to blocks $b_1,\ldots, b_\assoc$ results in a hit and leaves the state as it is. The only way to change the state is by inputting $b_{\assoc+1}$ and causing a miss, which results in the following state:
$$[b_{\assoc+1},b_1,b_2, \ldots,b_{\assoc-1}].$$
Again the only way to change the state is by inputting $b_\assoc$. After that it can only be changed by inputting $b_{\assoc-1}$, then $b_{\assoc-2}$, and so on until reaching again the initial state. By doing this we are just \emph{cycling} over the blocks, always evicting them in the same order. This produces $\assoc+1$ distinct reachable states.

If $\footprint\geq \assoc+1$ we prove we can reach every target state by inputting a sequence of blocks. Now there are two blocks outside the cache so we use one for cycling as before, having always an extra block outside. For illustration purposes we take the initial state $[b_1,b_2, \ldots,b_\assoc]$, and the target state $[b_\assoc, \ldots,b_2,b_1] $. We use an algorithmic approach to reach the target state from the initial one.  

We proceed from oldest to youngest block of the target state. First we cycle the cache until the second youngest target block, $b_2$, is evicted, we call this the \emph{stored} block. After that we continue to cycle the cache but without using the stored block, so it is not updated into the cache until we want to. When the oldest target block, $b_1$, is placed at the beginning of the cache, we input the stored block. This way we obtain the two oldest blocks of the target state in the target order $b_2,b_1$.

We cycle again until the third youngest target block, $b_3$, is evicted and cycle again until the second youngest $b_2$, is at the beginning. Then we access $b_3$ and obtain the three oldest blocks in the target order. We proceed in the same manner until the target state is reached. 
\end{Proof}

\begin{proposition}
\label{pro:abs_PLRU}
For $M_\PLRU$, the absorbed information for a filled cache is:

$$\ABS_\filledc(M_\PLRU,\footprint)=\begin{cases}
2^{\footprint-1}&\text{if } 1\leq\footprint\leq\assoc,\\
\frac{\footprint!}{(\footprint-\assoc)!}&\text{if }\footprint> \assoc.
\end{cases}$$

\end{proposition}
\begin{Proof}
Conceptually, PLRU maintains a binary tree with the blocks at the leaves and arrows on each non-leaf node pointing to one of the children, see Section~\ref{sec:replacement}.
Since the initial state is filled, all the victim's blocks are stored in some leaves of the tree.

If $\footprint\leq \assoc$, all the inputs are going to produce hits, and the only update is a permutation of the ages which, in the tree representation, is the flipping of the arrows away from the input block. Then, the number of reachable states depends on how many of these arrows can point in more than one direction. An arrow may point in any direction if its children both have at least one victim's block that can be used to flip the arrow.
So we we need to determine the number of non-leaf nodes where the two children have each at least one block.

Following the tree from root to leaves we can view the internal nodes with victim's blocks at the leaves reachable from both of its children as partitions of the set of victim's blocks into two subsets. This way, counting nodes with blocks in both children is equivalent to counting the number of partitions that can be performed until we obtain all singleton subsets. The amount of times we can partition a set of $\footprint$ elements into two subsets until obtaining singleton subsets is $\footprint-1$, independently of how the partitions are done.
Therefore, there are $\footprint-1$ arrows that may point to any direction. This produces $2^{\footprint-1}$ reachable states.

If $\footprint>\assoc$ we prove that every state is reachable from an initial one by inputting a sequence of blocks. For this we consider the target state as a tree with the blocks placed in specific leaves and the arrows in specific orientations.  We divide the proof in two parts, first the case where the target blocks are already in the initial state and later when they are not.

If the blocks in the initial state are the same as in the target state, that is, if for every subtree of the initial state there is a subtree in the target state with the same blocks (not considering the arrows), we have a different permutation of ages between the two. What we need to do then is obtain the correct permutation of the ages by inputting a sequence of blocks. This can be done by using only the blocks in the state as we see now. We recall that, from the case above, when $\footprint=\assoc$, the initial state has a victim block in every leaf. Therefore, there exists a sequence using only at most $\assoc$ blocks that reaches a state for every permutation of the arrows.

If the blocks in the initial state are not the same as in the target state, or if they are in wrong subtrees, we need to evict precise blocks from the tree and load them back in a different leaf in order to reach the target state. Again, since by only using blocks in the cache we can obtain any permutation of arrows, we can obtain a sequence that points the arrows to a specific block, evicts it and then modifies the arrows to load it back in a different leaf. Once the target blocks are in their corresponding leaves, we shift the arrows as before to obtain the target state. 
\end{Proof}

\subsubsection{Results for Empty Caches}
The case of an empty cache is more complex to analyze. First we need to explain a special behavior of PLRU that produces extra reachable states which increases its absorption with respect to the other two policies.
\begin{example}
\label{ex:PLRU_states}
Consider a 4-way cache  that starts in a state consisting of the attacker's blocks $\{x_0,x_1,x_2,x_3\}\subseteq B_a$ where we are going to access three victim blocks in a specific order, $a,b,c\in B_v$. For any of the three replacement policies the state becomes:
$$[x_0,x_1,x_2,x_3]
\underset{a}{\rightsquigarrow}
[a,x_0,x_1,x_2]
\underset{b}{\rightsquigarrow}
[b,a,x_0,x_1]
\underset{c}{\rightsquigarrow}
[c,b,a,x_0],$$
where the leftmost element of the lists has age zero and the one on the right is the oldest. Consider that we now access block $b$ again. The cache states transition to:
$[b,c,a,x_0]$ for LRU,
$[c,b,a,x_0]$ for FIFO and
$[b,c,x_0,a]$ for PLRU (note the age of the last attacker's block $x_0$).
The state obtained by PLRU is unreachable for the other two replacement policies, since they always fill up the cache consecutively from left to right. This illustrates how the information absorption for PLRU is larger than for the other policies.
\end{example}

The example is independent of the blocks being used but a consequence of the fact that we are inputting $k<\assoc$ blocks. For LRU and FIFO, any sequence using $k<\assoc$ victim blocks will transform an initial state $[x_0, x_1, \dots, x_{A-1}]$ to a state of the form $[\_, \dots, \_, x_0, \dots, x_{A-1-k}]$, where victim blocks are denoted by ``$\_$''. In the case of PLRU this is not always the case, as the previous example shows.

Following our definition of absorption, we assume that the victim may input any sequence of blocks.
Then the number of reachable cache states can be determined as follows:
\begin{compactenum}
	\item Determine the set of reachable \emph{configurations}, i.e., cache states in which the victim's memory blocks are not distinguished from each other, but instead represented by the \emph{placeholder} ``$\_$''.
	\item Determine for each configuration the number of concrete cache states the configuration represents, i.e., the number of ways the victim's blocks may fill its placeholders.
\end{compactenum}
This procedure can further be simplified upon by the following observation:
The number of concrete cache states that a configuration represents, only depends on its number of placeholders and the number of victim blocks to consider:
Given $k$ placeholders and $\footprint \geq k$ victim's memory blocks, a configuration represents exactly $\frac{\footprint!}{(\footprint-k)!}$ cache states.

Let $\Lambda_M(k, A)$ denote the number of reachable configurations under policy~$M$, associativity $A$, with exactly $k$ placeholders.
Accessing $\footprint$ distinct memory blocks may yield configurations with $0$ to $\footprint$ many placeholders.
Based on this notion, we obtain the following general characterization of a replacement policy's absorption:
\begin{proposition}\label{prop:last}
For any replacement policy $M$, the absorbed information starting from an empty cache is:
\begin{equation*}
\ABS_\emptyc(M,\footprint)=
\sum\limits_{k=0}^{\min\{\footprint, \assoc\}}\Lambda_M(k,\assoc)\frac{\footprint!}{(\footprint-k)!}.
\end{equation*}
\end{proposition}
\begin{lemma}
For LRU and FIFO, $\Lambda_M(k,\assoc)=1$ for any number of placeholders~$k$ and associativity $A$. 
For PLRU, $\Lambda_{M_\PLRU}(k,\assoc)$ is given by:
\begin{equation}
\label{eq:Lambda}
\Lambda_{M_\PLRU}(k,\assoc)=2 \cdot
\sum\limits_{i=\max\{1,k-\frac{\assoc}{2}\}}^{\min\{\frac{\assoc}{2},k-1\}} 
\Lambda_{M_\PLRU}(i,\tfrac{\assoc}{2})\cdot\Lambda_{M_\PLRU}(k-i,\tfrac{\assoc}{2}),
\end{equation}
if $1<k<A$ and $\Lambda_{M_\PLRU}(k,\assoc)=1$ if $k \leq 1$ or $k =\assoc$.
\end{lemma}
\begin{Proof}
For LRU and FIFO, since they fill the cache placing $k$ victim's blocks in the $k$ youngest ages, there is only one possible configuration for each value of $k$.

For PLRU we distinguish the three cases of $k$. If $k= 1$, there is only one block which is mapped to age 0. Repetitions of the same block do not modify the ages and so this is the only reachable configuration.

If $k =\assoc$, the state is completely filled with placeholders and so it is the only reachable configuration.

If $1<k<A$, we use the representation of PLRU caches as trees with blocks on the leaves and arrows on the non-leaf nodes. We consider that we input any sequence of blocks that ends up with exactly $k$ victim's blocks in the cache and study in which leaves these blocks can be. Note that we do not require the sequence to have exactly $k$ different blocks but rather at least $k$. Extra blocks may evict other victim's blocks and still end up with $k$ blocks in the cache.

We base the proof on the behavior of the root of the tree, its arrow and its two children. We study how different sequences of inputs affect them and use it to explain the elements of \eqref{eq:Lambda}. Since $\Lambda_{M_\PLRU}$ is constructed recursively, applying it to a child is equivalent to considering the child as tree on its own.

We first prove that any reachable state with exactly $k$ victim's blocks (and consequently placeholders) has at least 1 placeholder in each child and at most $\assoc/2$. The upper bound is trivial since it is the number of leaves in the child. For the lower bound consider the root of the tree. In the initial state the arrow points to one child. This child stores the first input $b_1$ of any sequence of victim's blocks, after which the arrow shifts to the other child. After this the sequence of inputs may have repetitions of $b_1$, which have no effect on the state, before inputting a new block $b_2$. Therefore, the child that does not store $b_1$ always stores $b_2$. 
This bounds can actually be reached by the sequence $b_1b_2b_1b_3b_1\ldots b_1b_k$. The repetitions of $b_1$ make the arrow always point to the same child  before inputting new blocks and so the blocks $b_2,b_3,\ldots,b_k$ are stored in the same child, with $b_1$ alone in the other.

Then, the number of configurations with $k$ placeholders depends on how many ways we can distribute them in the two children, constrained to at least 1 per child and at most $\assoc/2$. This corresponds to the limits of the sum in \eqref{eq:Lambda}; for each distribution of placeholders, we compute $\Lambda_{M_\PLRU}$ restricted to each child.

Finally, once a sequence has stored $k$ victim's blocks in the state, the distribution of placeholders is fixed. However, the sequence may input repetitions of the blocks stored in the cache, which  does not modify the number of placeholders but affects the arrow on the root, that shifts from one child to the other. Therefore, every distribution of placeholders accounts for two configurations, which produces the duplication in \eqref{eq:Lambda}.
\end{Proof}

\subsubsection{Comparison of Absorption}

Let us compare the absorption of LRU, FIFO, and PLRU based
on Propositions~\ref{prop:first} to \ref{prop:last}, for a cache set of
associativity 4. Similar results can be obtained for any associativity. The results depicted in Figure~\ref{fig:abs} can be obtained both from 
the formulas above or by simulation of caches.  We
highlight the following observations.
\begin{compactitem}
\item For each replacement policy, the absorbed information grows
  monotonically with the footprint, as expected.
\item The absorption for an empty initial state  is always larger than
  for a filled state.
\item For a filled initial state, LRU absorbs always at least as much
  information as the other replacement
  policies since every state is always reachable. For large enough footprints, the
  absorption coincides for all policies.
\item For an empty initial state
  PLRU absorbs most. This is due to the fact that PLRU may
  leave ``holes'' in the cache state, see Example~\ref{ex:PLRU_states}.
\item For a filled initial cache, FIFO does not absorb any
  information, whenever the footprint is smaller than the
  associativity.  This captures the intuition that preloading of
  sensitive data can increase security, as long as all data fits into
  the cache. In case it does not, the positive effect of preloading is, however,
  quickly undone.
\end{compactitem}

\begin{figure*}
\centering
\begin{tabular}{cc}
\begin{tikzpicture}
\pgfplotsset{every axis legend/.append style={font=\footnotesize }}
\begin{axis}[ymin=-0.5,
legend columns=3,
legend entries={FIFO,LRU,PLRU},
legend to name=abs,
width=0.575\textwidth,
xtick={0,1,2,3,4,5,6,7},height=5cm]

\AddCurve{Figures/abs_filled_fifo.txt}{color1}{square}{solid}
\AddCurve{Figures/abs_filled_lru.txt}{color2}{o}{solid}
\AddCurve{Figures/abs_filled_plru.txt}{color3}{diamond}{solid} 
 
\end{axis}
\end{tikzpicture}\customlabel{fig:abs_filled}{\ref*{fig:abs}a}&

\begin{tikzpicture}
\pgfplotsset{every axis legend/.append style={font=\footnotesize }}
\begin{axis}[ymin=-0.5,
width=0.575\textwidth,
xtick={0,1,2,3,4,5,6},height=5cm]

\AddCurve{Figures/abs_empty_fifo.txt}{color1}{square}{solid}
\AddCurve{Figures/abs_empty_lru.txt}{color2}{o}{solid}
\AddCurve{Figures/abs_empty_plru.txt}{color3}{diamond}{solid} 
  
\end{axis}
\end{tikzpicture}\customlabel{fig:abs_empty}{\ref*{fig:abs}b}\\
(a) Filled initial cache.&(b) Empty initial cache.
\end{tabular}
\ref*{abs}
\caption{Information absorption of a 4-way
  cache set. Figure~\ref{fig:abs_filled} depicts the case of a filled
  initial cache, part~\ref{fig:abs_empty} an empty one.  In both
  figures, the horizontal axis depicts the footprint, i.e., the number
  of memory blocks used. The vertical axis depicts the absorbed
  information on a logarithmic scale, that is, in {\em bits}. Note that in Figure~\ref{fig:abs_empty}, the line for LRU and FIFO coincides.}
\label{fig:abs}
\end{figure*}
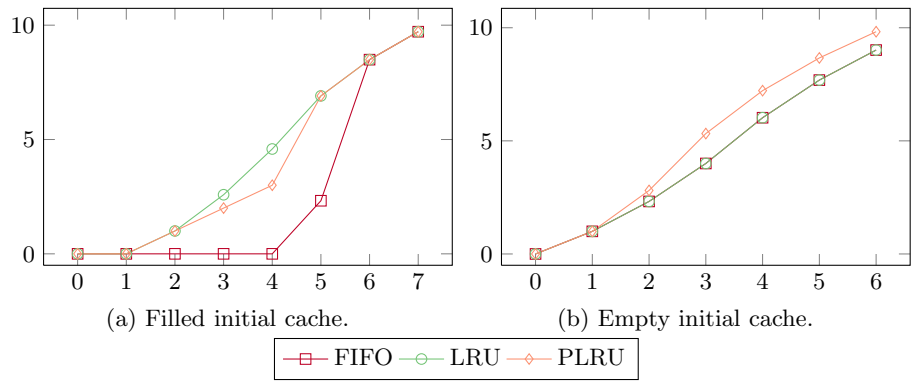

\section{Extraction of Information}
\label{sec:extraction}
In this section we characterize the information extraction for different cache replacement policies. That is, we characterize $\range{\randblock_a}$ from \eqref{eq:markov}. For this we develop a novel model that characterizes the information an adaptive attacker can learn about the initial state of a Mealy machine. We then use the model to derive bounds on the information that can be extracted from caches with different replacement policies.

\subsection{Probing Strategies}
Let $M=(\keys,\queries,\obs,\up,\view)$ be a Mealy machine. A {\em probe} $p$ of $M$ is an alternating sequence $p=\queriesval_1\obsval_1\queriesval_2\dots\queriesval_n\obsval_n$ of inputs $\queriesval_i\in\queries_a\subseteq\queries$ and observations $\obsval_i\in \obs$, such that $M$ outputs $\obsval_1\dots\obsval_i$ when the sequence $\queriesval_1\dots\queriesval_i$ is the input. We say that a state $s\in\keys$ is {\em coherent} with probe $p$ if, for all $i\in\{1,\dots,n\}$, we have
\begin{equation*} 
\view_{\queriesval_i}\up_{\queriesval_{i-1}}\cdots\up_{\queriesval_1}(s)=o_i\ , 
\end{equation*} 
i.e., the probe does not exclude $s$ as a potential initial state of $M$. Along the lines of~\cite{AskarovS07,koepfbasin07}, we define the adversary's {\em knowledge set} $K(p)$ about the initial state of $M$ as the subset of possible states that are coherent with probe $p$.  
\begin{equation*} 
K(p)=\{s\in \keys_v\mid s\text{ is coherent with } p\}
\end{equation*} 
For convenience, we also define the adversary's {\em final knowledge set} $FK(p)$ as the set of states that $M$ may be in after receiving the inputs and producing the outputs in the probe $p$:
\begin{equation*}
FK(p) = \{\up_{\queriesval_{n}}\cdots\up_{\queriesval_1}(s) \mid s \in K(p)\}
\end{equation*}

An adversary may be able to choose inputs based on previous observations, that is, the probing can be adaptive. To model adaptivity we introduce probing strategies. A {\em probing strategy} is a function from a sequence of observations to an input symbol, $\att:\obs^*\mapsto\queries_a$. This way, the first input to make comes from applying the function to the empty sequence, $\queriesval_1=\att(\varepsilon)$, the second input is a function of the previous observation, $\queriesval_2=\att(\obsval_1)$, and so, for any $i$ $\queriesval_i=\att(\obsval_1\ldots\obsval_{i-1})$.
We say that $p$ is a probe of $\att$, if $p$ may be obtained from the probing strategy~$\att$.

We now present a toy example that we will use through the section to illustrate the use of probing strategies.
\begin{example}\label{ex:probing_strategy}

Consider a Mealy machine where $\keys=\keys_v=\queries_a=\{0, 1, \ldots, 6\}$, the observation and transition function are:

$$\view_\queriesval(\keysval)=\begin{cases}
0&\text{if }\keysval<\queriesval-1,\\
2&\text{if }\keysval\in[\queriesval-1,\queriesval+1],\\
1&\text{if }\queriesval+1<\keysval.
\end{cases}\quad
\up_\queriesval(\keysval)=\begin{cases}
\keysval+1&\text{if }\keysval<\queriesval,\\
\keysval&\text{if } \keysval\in[\queriesval,\queriesval+1],\\
\keysval-1&\text{if }\queriesval+1<\keysval.
\end{cases}$$
Consider the probing strategy given by the function $\att(\obsval_1\ldots\obsval_n)=0+\sum_{i=1}^n\obsval_i$, which starts by inputting $0$ and determines the next input based on the previous outputs. We will later see that $\att$ is a good probing strategy in this example.
\end{example}

By definition, we can apply a probing strategy indefinitely on sequences of arbitrary length and thus probe the Mealy machine indefinitely. 
However, at some point additional inputs are of no use, as the following definition characterizes.
\begin{definition}
\label{def:indpart}
We say that a probe $p=\queriesval_1\obsval_1\queriesval_2\dots\queriesval_n\obsval_n$ of probing strategy $\att$ is \emph{depleted} w.r.t. to \att, if for all probes $q$ of $\att$ that are extensions of $p$, i.e., $q = p \queriesval_{n+1}\obsval_{n+1}\queriesval_{n+2}\dots\queriesval_m\obsval_m$, the knowledge sets are equal, i.e., $K(p) = K(q)$.
We say a depleted probe $p=\queriesval_1\obsval_1\queriesval_2\dots\queriesval_n\obsval_n$ is of \emph{minimal length} when, a probe $q$ made of a sub-sequence of it, $q=\queriesval_{k_1}\obsval_{k_1}\queriesval_{k_2}\dots\queriesval_{k_i}\obsval_{k_i}$ for any $i<n$, is not depleted.
\end{definition}

We next show that the knowledge sets of depleted probes of a probing strategy form a partition of the states of $M$. That is, the knowledge sets of distinct sequences are pairwise disjoint and their union contains all states.
\begin{proposition} 
Given a probing strategy $\att$, the set of all knowledge sets produced by depleted probes w.r.t. $\att$
$$\partition_\att=\{K(p)\mid \textit{probe } p=\att(\varepsilon)\obsval_1 \ldots \att(\obsval_1 \ldots \obsval_{n-1})\obsval_n \wedge  p \textit{ is depleted w.r.t. } \att\},$$
is a partition of the set of possible states $\keys_v$.
\end{proposition}
\begin{Proof} 
We will first prove the following related statement:
Let $\partition_\att(i)$ be defined as follows:
$\partition_\att(i)=\{K(p) \mid \textit{probe } p=\att(\varepsilon)\obsval_1 \ldots \att(\obsval_1 \ldots \obsval_{i-1})\obsval_i\}$.
Then $\partition_\att(i)$ is a partition of $\keys_v$ for all $i$.

We prove this statement by induction on $i$.\\
Induction basis: For $i=0$, $\partition_\att(i) = \{K(\epsilon)\}$, and $K(\epsilon) = \keys_v$. So $\partition_\att(i)$ is a trivial partition of $\keys_v$.\\
Induction step: For $i+1$, $\partition_\att(i+1) = \{K(p\sigma_{i+1}\obsval_{i+1}) \mid \textit{probe } p=\att(\varepsilon)\obsval_1 \ldots$ $\att(\obsval_1 \ldots \obsval_{i-1})\obsval_{i} \wedge \sigma_{i+1} =\att(\obsval_1 \ldots \obsval_{i}) \wedge \obsval_{i+1} \in \view_{\sigma_{i+1}}(FK(p))  \}$.
By induction hypothesis $\partition_\att(i)$ is a partition of $\keys_v$.
It is easy to see that for each probe $p$ of length $i$, $K(p)$ is partitioned by $\{K(p\sigma_{i+1}\obsval_{i+1}) \mid \sigma_{i+1} =\att(\obsval_1 \ldots \obsval_{i}) \wedge \obsval_{i+1} \in \view_{\sigma_{i+1}}(FK(p))\}$.
So $\partition_\att(i+1)$ is a refinement of $\partition_\att(i)$ and thus also a partition of $\keys_v$.

Let $n$ be the length of the longest depleted probe of minimal length.
Then $\partition_\att = \partition_\att(n)$ as all probes considered in $\partition_\att(n)$ must be depleted, and as extensions of depleted probes have the same knowledge set as their corresponding depleted probes of minimal length. 
\end{Proof}

Before starting the probing, the attacker knows that the victim's state is an element of the set $\keys_v$. As he makes inputs and refines the knowledge sets, he reduces the number of coherent states and thus learns information about the victim's initial state.
As depleted probes correspond to unrefinable knowledge sets, there is no point in further queries once a probe is depleted.

When constructing a strategy, the attacker needs to consider all the possible outputs that he might observe when eventually applying his strategy. 
Once all the knowledge sets obtained from an attack strategy cannot be further refined by additional queries, the probes are depleted and the attacker has along the way obtained the finest partition of the set $\keys_v$ under that strategy and all possible extensions.

\begin{table}
\setlength{\tabcolsep}{7pt} 
\centering \footnotesize
\begin{tabular}{|ccccccccccccc|} \hline
\quad 0/0&&1/1&\textbf{0}&2/2&&3/3&&4/4&&5/5&&6/6 \quad\ \\
\hline
\quad 0/0&\textbf{2}&1/1&\vline &2/1&&3/2&\textbf{1}&4/3&&5/4&&6/5 \quad\ \\
\hline
\quad 0/1&\vline &1/2&\vline &2/1&\textbf{3}&3/2&\vline &4/2&&5/3&\textbf{2}&6/4 \quad\ \\
\hline
\quad 0/1&\vline &1/2&\vline &2/2&\vline &3/3&\vline &4/2&\textbf{4}&5/3&\vline &6/3 \quad\ \\
\hline
\quad 0/1&\vline &1/2&\vline &2/2&\vline &3/3&\vline &4/3&\vline &5/4&\vline &6/3 \quad\ \\
\hline
\end{tabular}\\
\vspace{1ex} \caption{Partition from Example~\ref{ex:partition}. } \label{tab:partition}
\end{table}

\begin{example} \label{ex:partition} 
Following Example~\ref{ex:probing_strategy} we apply the probing strategy to the set of possible states and obtain the partition shown in Table~\ref{tab:partition}. Each row shows the knowledge sets before and after the elements are updated (left and right, respectively). The first row shows the initial knowledge set, i.e., $\keys_v$. The bold face 0 indicates the first input symbol, which partitions the initial knowledge set into two knowledge sets, corresponding to the two possible outputs of the Mealy machine on the input $0$.
For each resulting knowledge set, except for the singleton ones where the probes are depleted, the figure then indicates the next input following the probing strategy and how it partitions its knowledge set.  After at most four inputs we obtain a partition of all singleton knowledge sets.
\end{example}

For every attack strategy there is a finite set of depleted probes of minimal length. 
We define $\randblock_a=\randblock_\att$ from \eqref{eq:markov} as the random variable that captures the sequence of observations obtained when following probing strategy $\att$ until obtaining a depleted probe of minimal length.
So $\range{\randblock_\att}\subseteq\obs^*$  is the set of sequences of observations obtained from the depleted probes of minimal length of $\att$. Every depleted probe corresponds to a knowledge set; so we can relate every element of $\range{\randblock_\att}$ to a knowledge set. Therefore, computing $|\range{\randblock_\att}|$ is equivalent to counting the number of knowledge sets in the partition induced by the strategy $\att$.

\begin{definition}\label{def:optstrat}
We say that a strategy $\att$ is {\em optimal} if the
partition $\partition_\att$ it induces on a set of possible states $\keys_v$, has the maximal number of knowledge sets among all
strategies. We call this number $\maxleak$ the \emph{maximum information leakage}.
\end{definition}

The strategy presented in Example~\ref{ex:probing_strategy} is actually optimal since no partition can be better than the one that produces singleton knowledge sets. On the other hand, the strategy $\att(\obsval_1\ldots\obsval_n)=1+\sum_{i=1}^n\obsval_i$ is not optimal since the first input, $1$, is not able to distinguish the initial states $0$ and $1$, which are both updated to $1$ as a result of the input, $\up_1(0)=\up_1(1)=1$, and so they can not be distinguished by this strategy.

\subsection{Information Extraction in Caches}
\label{sec_extraction_apndx}
Here we derive bounds on the maximum information leakage for the three replacement policies. We prove bounds for LRU and FIFO based on the associativity of the cache and prove that for PLRU this bound depends also on the footprint.

\ifshowproof

\textbf{Notation.} Given a set of cache states $\cache$ we use the following shortcuts: $\cache(b)=\{c(b)\mid c\in\cache\}$, $\up_{b_0\ldots b_{n-1}}(\cache)=\{\up_{b_0\ldots b_{n-1}}(c)\mid c\in\cache\}$ and $\view_b(\cache)=\{\view_b(c)\mid c\in\cache\}$.

\begin{definition}
\label{def:deterministic}
We say that a set of cache states $\cache$ has $n\leq\assoc$ \emph{deterministic ages} if all the states in $\cache$ have the same $n$ youngest blocks. That is, if there exists a sequence of blocks $a_0\ldots a_{n-1}$ such that $\cache(a_i)=\{i\}$ for all $i\leq n-1$.
\end{definition}

\begin{lemma}
\label{cor:deterministic_depleted}
Let $\cache$ be a set of cache states of associativity $\assoc$ and let $p$ be a probe. We have that $p$ is depleted if and only if  $FK(p)$ has $\assoc$ deterministic ages.
\end{lemma}
\begin{proof}
Since $p$ is depleted we have that $\sizeof{FK(p)}=1$. Then $FK(p)$ trivially has $\assoc$ deterministic ages.

If $FK(p)$ has $\assoc$ deterministic ages, all the blocks are mapped to the same age for every state of $FK(p)$ (either the deterministic ages or age $\assoc$). This means that $\sizeof{(FK(p))}=1$ which implies that it is unrefinable and so $p$ is depleted. 
\end{proof}

\begin{lemma}
\label{lem:repetition}
Consider a set of cache states $\cache$ of associativity $\assoc$ that uses either LRU or FIFO. Inputting a block mapped to a deterministic age has no effect on the number of deterministic ages.
\end{lemma}
\begin{proof}
Assume that $\cache$ has $n$ deterministic ages. Then there exists a sequence of blocks $a_0\ldots a_{n-1}$ such that $\cache(a_i)=\{i\}$ for all $i\leq n-1$. Any input with a block $a_j\in\{a_0,\ldots ,a_{n-1}\}$ results in a hit. For LRU, following \eqref{eq:lru_permutation}, the new ages are:
$$\up_{a_j}(\cache)(a_i)=
\begin{cases}
  \{0\} & \text{if } i = j\\
  \{i + 1\} & \text{if } i < j\\
  \{i\} & \text{if } i>j
\end{cases} $$
so the blocks $a_0\ldots a_{n-1}$ are still mapped to the first $n$ ages for all states which results in $n$ deterministic ages. For FIFO, since hits do not reorder blocks \eqref{eq:fifo_permutation}, the conclusion is the same. 
\end{proof}

\begin{lemma}
\label{lem:deterministic}
Consider a set of cache states $\cache$ of associativity $\assoc$ and a sequence of $n$ inputs $b_0\ldots b_{n-1}$ with $b_i\neq b_j$ for any $i\neq j$.

If the cache  uses LRU, for any $i\leq n$, we have that $\up_{b_0\ldots b_{n}}(\cache)(b_i)=\{n-i\}$.

The same result holds for FIFO if each of the inputs results in a miss, i.e., $\up_{b_0\ldots b_{i-1}}(\cache)(b_i)=\{\assoc\}$ for all $i\leq n-1$.
\end{lemma}
\begin{proof}
We proceed by induction. When $n=0$, $\up_{b_0}(\cache)(b_0)={0}$; LRU places the block in the beginning for both hits \eqref{eq:lru_permutation} and misses \eqref{eq:cache_update} and FIFO does it for misses \eqref{eq:cache_update}, by assumption this is the case for all inputs.

We assume the hypothesis is true for $n$ and input $b_{n+1}$. This block is distinct from the previous ones so, by the induction hypothesis, $\up_{b_0\ldots b_{n}}(c)(b_{n+1})>n$ for any $c\in\cache$, i.e. $b_{n+1}$ is older than the previous. For FIFO its age is actually $\up_{b_0\ldots b_{n}}(\cache)(b_{n+1})=\{\assoc\}$. Then $b_{n+1}$ is placed at age zero, $\up_{b_0\ldots b_{n+1}}(\cache)(b_{n+1})=\{0\}$, and the others increase their ages by one, $\up_{b_0\ldots b_{n+1}}(c)(b_{i})=\up_{b_0\ldots b_{n}}(c)(b_{i})+1=n+1-i$ for all $c\in\cache$ and $i\leq n$. 
\end{proof}

Assume now that we have a set of cache states $\cache$ and we input a sequence $b_0\ldots b_{k-1}$ with the requirements given in Lemma~\ref{lem:deterministic}. Following this Lemma, $\up_{b_0\ldots b_{k-1}}(\cache)$ has $k$ deterministic ages with $a_i=b_{k-1-i}$ from Definition~\ref{def:deterministic}. Now consider that we extend the sequence of inputs with $b_{k}\ldots b_{m}$ with $b_i\in\{b_0,\ldots ,b_{k-1}\}$ for $k\leq i\leq m$, that is, with blocks already mapped to deterministic ages. Then following Lemma~\ref{lem:repetition} $\up_{b_0\ldots b_{m}}(\cache)$ still has $k$ deterministic ages.

If we continue to extend the sequence of inputs with new blocks (provided they produce misses for FIFO) we will produce extra deterministic ages on the updated set of states. If we extend with blocks already mapped to deterministic ages, the number of deterministic ages is not modified.
\fi

We consider two types of attackers in terms of their set of memory blocks.
\begin{compactitem}
\item\emph{Shared memory attacker}. The  attacker's set of blocks includes the victim's ones, $ B_v\subset B_a$.
\item \emph{Disjoint memory attacker}. The sets of blocks of the attacker and the victim are disjoint $ B_v\cap B_a=\emptyset$. 
\end{compactitem}

\begin{proposition}
\label{pro:extraction_LRU}
Consider $M_\LRU$ and $M_\FIFO$ with associativity $\assoc$ and a shared memory attacker. The maximum information leakage on any set of states is bounded by $2^\assoc$ for $M_\LRU$ and by $(\assoc+1)!$ for $M_\FIFO$.
\end{proposition}
\begin{Proof}
\emph{LRU. }
Consider any given strategy $\att$ and any probe produced by $\att$, $p=b_1o_1\ldots b_no_n$ with $n\geq \assoc$ and $|\{b_1,\ldots ,b_n\}|=\assoc$. Following Lemmas~\ref{lem:repetition} and~\ref{lem:deterministic}, since $b_1\ldots b_n$ is formed with alternating sub-sequences of new blocks and repetitions of them, $FK(p)=\up_{b_1\ldots b_{n}}(K(p))$ has $\assoc$ deterministic ages which, by Lemma~\ref{cor:deterministic_depleted}, means that $p$ is depleted. Then any probe is depleted w.r.t. $\att$ if it has $\assoc$ different inputs.

We now prove that, for any probe given by $\att$, repetitions of inputs do not partition the knowledge sets. Given the non-depleted probe $p=b_1o_1\ldots b_no_n$, for any value of $n$, we have that a new input $b$ produces $\sizeof{\view_{b}(FK(p))}=1$ if $b= b_i$ for some $i\leq n$ since $b_i$ is mapped to a deterministic age. For any other input, the view function is trivially bounded by 2, $\sizeof{\view_{b}(FK(p))}\leq 2$.

Given the set of possible states $\keys_v$, the first input given by $\att$ will produce at most two knowledge sets. For each of these knowledge sets, the second input will partition them into two knowledge sets, making a total of up to four knowledge sets, unless it is a repetition, in which case there is no partition. We can partition the knowledge sets further until the probes are depleted, which happens after $\assoc$ different inputs. Each not repeated input at most doubles the amount of knowledge sets so, after $\assoc$ inputs the strategy produces up to $2^\assoc$ knowledge sets. 

\emph{FIFO. }
Consider any given strategy $\att$ and any probe produced by $\att$, $p=b_1o_1\ldots b_no_n$ with $n\geq \assoc$ and $\assoc$ misses. Following Lemmas~\ref{lem:repetition} and \ref{lem:deterministic}, since $b_1\ldots b_n$ is formed with alternating sub-sequences of new blocks and repetitions of them, $FK(p)$ has $\assoc$ deterministic ages which, by Lemma~\ref{cor:deterministic_depleted}, means that $p$ is depleted. Then any probe is depleted w.r.t. $\att$ if it has $\assoc$ misses.

Consider a non-depleted probe $p=b_1o_1\ldots b_no_n$, for any value of $n$, that we extend with an input $b$ such that $\view_b(FK(p))=\{\hit\}$. Then $FK(pb\hit)=FK(p)$, that is, obtaining a hit for all the states has no effect on the partition. An input $b$ can return a hit for all the states in a final knowledge set in two cases:
\begin{compactenum}
\item The probe, starting from the last miss, is one of the form $p=b_1\miss b_2\hit\ldots b_n\hit b\hit$ with $b=b_i$ for some $i\leq n$.
\item The input $b$ is mapped to a deterministic age.
\end{compactenum}

We now have a characterization of how depleted probes of minimal length look like. First, all depleted probes have $\assoc$ misses but may have a different number of hits between each miss. Second, in order for the probes to be of minimal length, their inputs do not return a hit for all the states in the current final knowledge set. 

Since new misses introduce new deterministic ages, the maximum number of consecutive non-repetitive hits is reduced by one with each miss. Before any miss, $\att$ can produce up to $\assoc$ hits with these restrictions, which gives up to $\assoc+1$ different probes. After the first miss, $\att$ extends each probe with up to $\assoc-1$ consecutive hits which makes up to $(\assoc+1)\assoc$ different probes. In the end, any attack strategy can produce up to $(\assoc+1)!$ depleted probes of minimal length and the same amount of knowledge sets.
\end{Proof}

\begin{proposition}
\label{pro:extraction_PLRU}
Consider $M_\PLRU$ with associativity $\assoc\geq 4$\footnote{Note that for associativity $2$, \PLRU\ and \LRU\ coincide.} and a shared memory attacker. Let $\maxleak(\footprint)$ be the maximum information leakage obtained with a given footprint $\footprint \geq \assoc$. It holds that $\maxleak(\footprint+1)\geq\maxleak(\footprint)+1$.
\end{proposition}
\begin{Proof}
Let $\att$ be an attack strategy that obtains $\maxleak(\footprint)$ given a set of possible states. We are going to prove that one empty knowledge set from $\partition_\att$ when using $\footprint$ blocks is non-empty given an extra memory block i.e., when using $\footprint+1$ blocks.

Consider a probe where the first $\assoc-1$ observations are misses and the corresponding final knowledge set $\keys_1$; all the states in this set have the same blocks mapped to the younger ages and an unknown one in age $\assoc-1$ that we call $b'$. Consider the set of yet unused victim's blocks $ B'_v\subset B_v$, $b'\in B'_v$; any input from that set evicts $b'$ in the case of a miss. But before we input from $ B'_v$ we input one previously-used block and update the states so that $b'$ is mapped to age $\assoc-2$. We use the block in age 1 to modify the age of $b'$ since $\Pi^{\mathit{PLRU}}_{1}(\assoc-1)=2\cdot\Pi^{\mathit{PLRU}}_{0}(\assoc/2-1)=2\cdot(\assoc/2-1)=\assoc-2$, see \eqref{eq:plru_permutation}. In order to do this, the associativity must be at least $4$ so that $\assoc-1>1$.

After this, a new input from $ B'_v$ does not evict $b'$ in the case of a miss. Therefore, the input partitions $\keys_1$ into a knowledge set with $b'$ and another knowledge set without $b'$ but an unknown block in age $\assoc-1$. Repeating this process of placing the block from age $\assoc-1$ in age $\assoc-2$ and later inputting a new block allows to partition the knowledge set that returned a miss, without evicting an unknown block in the case of a miss, and therefore maximize the number of knowledge sets.

There is an attack strategy $\att$ that follows this process in order to obtain $\maxleak(\footprint)$. Then the knowledge set of the probe where all the observations are misses is empty. Now suppose that we have an extra memory block. Then, the knowledge set of the probe where all the observations are misses is not empty as it contains states with the extra block in age $\assoc-1$. This way we have increased the maximum information leakage by one. 
\end{Proof}

In the case of associativity four for $M_\PLRU$ the maximum information leakage is increased by eight with every new memory block, this can be seen in Figures \ref{fig:plru_filled} and \ref{fig:plru_empty}. This result also implies that the maximum information leakage for PLRU is unbounded.

\begin{proposition}
\label{pro:extraction_LRU_disj}
Consider $M_\FIFO$ and $M_\LRU$ with associativity $\assoc$, and a disjoint memory attacker. The maximum information leakage on any set of states is bounded by $\assoc+1$.
\end{proposition}
\begin{Proof}
We base the proof on the state of the cache before the victim accesses it. There are $\assoc$ attacker's blocks  $x_0,\ldots ,x_{\assoc-1}$ such that $c(x_i)=i$. 

Now we make use of the fact that for both, $M_\FIFO$ and $M_\LRU$, for any block $b'$, $c(b)=upd_{b'}(c)(b)$ if $c(b)>c(b')$, see \eqref{eq:fifo_permutation} and \eqref{eq:lru_permutation}. This means that, when the attacker inputs one block he does not modify the ages of older ones. This is a way to probe the cache without evicting any blocks.

Assume the attacker inputs $x_0$; he gets a hit or a miss and partitions $\keys_v$ accordingly. The knowledge set that returned a miss has states with zero attacker's blocks, all inputs with attacker's blocks return the same output so it can not be partitioned further. The other knowledge set has at least one attacker's block and, following the property stated above, the older blocks, i.e. $x_1, x_2,$ etc, have not been evicted. Now the attacker inputs $x_1$ and partitions the set into the states with exactly one attacker's block ($x_0$) and the ones with at least two.

Following this sequence $x_0,x_1,\ldots $, every input singles out one unrefinable knowledge set but does not affect future inputs. In the end the attacker produces up to $\assoc+1$ knowledge sets. 
\end{Proof}

\begin{proposition}
\label{pro:extraction_PLRU_dis}
Consider $M_\PLRU$ with associativity $\assoc$, footprint $\footprint$, and a disjoint memory attacker. The maximum information leakage is bounded by \linebreak$\sum_{k=0}^{\footprint} \Lambda_\PLRU(k,\assoc)$ where $\Lambda_\PLRU(k,\assoc)$ is defined as in \eqref{eq:Lambda}.
\end{proposition}
\begin{Proof}
The information extraction is intuitively bounded by the number of configurations of attacker's and victim's blocks that a disjoint-memory attacker can distinguish. For each value of $k\in[0,\footprint]$, $\Lambda_{M_\PLRU}(k,\assoc)$ gives the number of possible configurations using $k$ victim's blocks, leaving $\assoc-k$ attacker's blocks. Therefore, by summing all of them up to the used footprint, we obtain the total number of configurations. Note that this bound may not be tight.
\end{Proof}

\section{An Algorithm for Information Extraction}
\label{sec:algorithm}

In this section we present an algorithm for computing the maximum
information leakage $\maxleak$ for a given Mealy machine. The
algorithm complements
Propositions~\ref{pro:extraction_LRU} to \ref{pro:extraction_PLRU_dis} in
that it can deliver $\maxleak$ for a specific set of states
$\keys_v\subseteq \keys$ and an arbitrary Mealy machine. 
We use it later to compute extraction
w.r.t. a given memory footprint, and to replace the engine for counting
cache states in the CacheAudit static analyzer, leading to tighter
bounds on the leakage.

In principle, our algorithm enumerates all attack strategies $\att$
and computes their partitions $\partition_\att$ by grouping states in
$\keys_v$ according to the corresponding
observations. Additionally, we use two techniques for improving
efficiency and ensuring termination: 
\begin{asparaitem}
\item First, instead of maintaining the knowledge sets $K(p)$, for
  every probe $p$, we maintain the final knowledge set $FK(p)$. Using
  the final knowledge set enables us to track the number of original
  knowledge sets, as required for computing leakage. At the same time
  it enables re-use of the computation leading to $FK(p)$ across
  different strategies.
\item Second, we need to identify cycles when refining partitions in
  order to ensure termination. We say that a probe $q$ is {\em
    redundant} w.r.t another probe $p$, if $FK(pq)=FK(p)$. That is,
  the probe $q$ does not further refine the (final) knowledge set of
  $p$. The probe $q$ represents a cycle, which we detect by keeping
  track of already visited final knowledge sets.
\end{asparaitem}
The pseudocode is given in Algorithm~\ref{algorithm}. We next argue its correctness.

\begin{algorithm}[h!]
\SetKwFunction{Part}{Partition}
%\SetKwProg{Function}{Function}{:}{end}
\Part{$\keys,\view,\up,\queries^a,\mathcal{S}$}{
  \KwData{set of possible states $\keys$ (initially $\keys=\keys_v$), observation function $\view$, transition function $\up$, set of attacker's inputs $\queries^a$, flag sets $\mathcal{S}$ (initially $\mathcal{S}=\emptyset$).}
  \KwResult{number of knowledge sets $\maxleak$ in the partition.}
  \Begin{
  \tcp{Look for redundant sequences}
  \If{$\keys\in\mathcal{S}$\nllabel{ln:redundant}}{\Return 1\;}
  $\maxleak=1$\;
  \ForEach{$\queriesval\in\queries^a$}{\nllabel{ln:queries}
    \tcp{If the leakage is equal to the size of the set, finish}
	\If{$\maxleak=|\keys|$\nllabel{ln:stop}}{\Return $\block_{max}$\;}
    \tcp{If the partition is not refined save the set}
    \uIf{$|\view_\queriesval(\keys)|=1$\nllabel{ln:flag}}
    {$\mathcal{S}'=\mathcal{S}\cup\{\keys\}$\;}
	\tcp{If the partition is refined erase the saved sets}
	\Else{$\mathcal{S}'=\emptyset$\;}
	\ForEach{$\obsval_i\in \view_\queriesval(\keys)$\nllabel{ln:obs}}{
	  $\keys_i=\{\keysval\in\keys \mid \view_\queriesval(\keysval)=\obsval_i\}$\tcp*{partition}\nllabel{ln:part}
	  $\keys_i'=\up_\queriesval(\keys_i)$\tcp*{update}\nllabel{ln:up}
	  $\block_i$ = \Part{$\keys_i',\view,\up,\queries^a,\mathcal{S}'$}\tcp*{recursion}\nllabel{ln:recurs}
	}
	\tcp{Increase the number of produced knowledge sets}
	$\maxleak=\max(\maxleak,\sum_i\block_i$)\;\nllabel{ln:max}
  }
  \Return $\block_{max}$\;\nllabel{ln:optimal}
}
}
\caption{Partition function.}
\label{algorithm}
\end{algorithm}

\begin{proposition}
Given a Mealy machine $M=(\keys,\queries,\obs,\up,\view)$, Algorithm~\ref{algorithm} terminates and finds the maximum information leakage $\maxleak$ for a set of possible states $\keys_v$.
\end{proposition}
\begin{Proof}
  The algorithm recursively studies all the possible sequences of
  inputs. In every call, it cycles through all the inputs
  (line~\ref{ln:queries}) and for each output of the observation
  function (line~\ref{ln:obs}), partitions the set into final
  knowledge sets (line~\ref{ln:part}), updates every final knowledge
  set (line~\ref{ln:up}) and recursively calls again the function
  (line~\ref{ln:recurs}). Then, each execution obtains from the
  following calls the best way to partition the final knowledge set
  for each input (line~\ref{ln:recurs}), chooses the one with the
  maximum value (line~\ref{ln:max}), and returns it to the previous
  level (line~\ref{ln:optimal}).  Our algorithm terminates if the number of knowledge sets
  is equal to the size of the set of possible states
  (line~\ref{ln:stop}), at which point the knowledge sets cannot be
  further refined.

The flag sets $\mathcal{S}$ are used to keep track of the redundant
sequences. Whenever an input does not partition the set (i.e. there is
only one observation), it is saved (line~\ref{ln:flag}). If a later
call sees that the updated set $\keys$ is equal to one saved in
$\mathcal{S}$ then that call has produced a redundant sequence
and so it is stopped with one final knowledge set
(line~\ref{ln:redundant}). This procedure also guarantees that the
algorithm terminates. If all the inputs that only produce one
observation result in a redundant sequence, the algorithm is forced
to choose one that partitions the set and eventually depletes the
probes. Once a probe is depleted, the algorithm does not extend it 
anymore since every extension of a depleted probe is redundant.
\end{Proof}

\section{Experimental Results}\label{sec:experiments}

\subsection{Extraction (Program-independent)}

We use two alternative approaches for the program-independent
evaluation of extraction properties cache replacement policies. The
first is to rely on the upper bounds of
Propositions~\ref{pro:extraction_LRU} to
\ref{pro:extraction_PLRU_dis}. The second is to apply the algorithm
presented in Section~\ref{sec:algorithm} to a set of states that
represent the absorbed information for a given footprint. We determine
that set for each cache replacement policy by a simple fixpoint
computation. This algorithmic approach is more precise because it
takes the absorbed information as a baseline, but it comes at the
expense of higher computational cost.

We obtain the following results by using Algorithm~\ref{algorithm}, where we consider a single 4-way
cache set. Figure~\ref{fig:extraction} depicts our data. We highlight
the following results:
\begin{asparaitem}
\item For shared-memory adversaries, FIFO and LRU reach the bound on
  the maximum information leakage given in
  Proposition~\ref{pro:extraction_LRU}, which is independent of the
  footprint, see Figures~\ref{fig:fifo_filled} to
  \ref{fig:lru_empty}. In contrast, with PLRU the number of knowledge sets increases
  with the footprint as predicted by
  Proposition~\ref{pro:extraction_PLRU}, see Figures~\ref{fig:plru_filled} to \ref{fig:plru_empty}.
\item For disjoint-memory adversaries and a filled initial state we
  always obtain zero leakage. For PLRU and a footprint of 2 or 3 some cache lines remain unoccupied. As before,
  these unoccupied lines trigger additional observations, which
  explain the bump in Figure~\ref{fig:plru_filled}.
\item We observe that FIFO exhibits the smallest difference between
  absorption and extraction among all policies, i.e. once absorbed, it
  is comparably easy to extract information from the cache, see Figures~\ref{fig:fifo_filled} to \ref{fig:fifo_empty}. This is
  because FIFO does not reorder blocks upon hits, which makes systematic search
  for the cache state easier.
\end{asparaitem}

\subsection{Extraction (Program-dependent)}

We now use Algorithm~\ref{algorithm} for computing the information
that can be extracted from the cache state w.r.t. a specific
program. For this, we use as a basis the set $\keys_v$ of states
output by the CacheAudit static analyzer, when run on an
implementation of AES 256. In this example we use a cache consisting
of several independent cache sets of associativity 4, blocks of 64
bytes and overall sizes of 4, 8, and 16 KB. We consider two cases, one
that starts from a filled cache and one that starts from an empty cache.
 
The full results are given in Figure~\ref{fig:aes}; here we highlight
the following results.

\begin{asparaitem}
\item We obtain the bounds on the absorbed information corresponds to
  using the CacheAudit static analyzer. The difference between the
  absorbed information and the extractable information corresponds to
  the precision gained by the development in this paper. This gain is
  generally higher when sets contain more blocks, and reaches up to 50
  bits for LRU on a 4K cache with empty initial state and a shared memory attacker, see
  Figure~\ref{fig:lru_aes_empty}. That is, our extraction algorithm is
  a simple but powerful replacement for the model counting algorithms
  in CacheAudit.
\item The figures show a change in slope at different points. This is
  due to the fact that the leakage about the full cache state is
  computed as the product of the leakages about the individual
  sets. When increasing the cache size for a fixed program, the
  footprint in each of the sets reduces. The combined effect of
  considering more sets with smaller footprint each accounts for the
  change in slope.
\end{asparaitem}

\pgfplotsset{every axis legend/.append style={font=\footnotesize }}
%---------------------------------------------------------------------------------
%Results for independent analysis
%---------------------------------------------------------------------------------
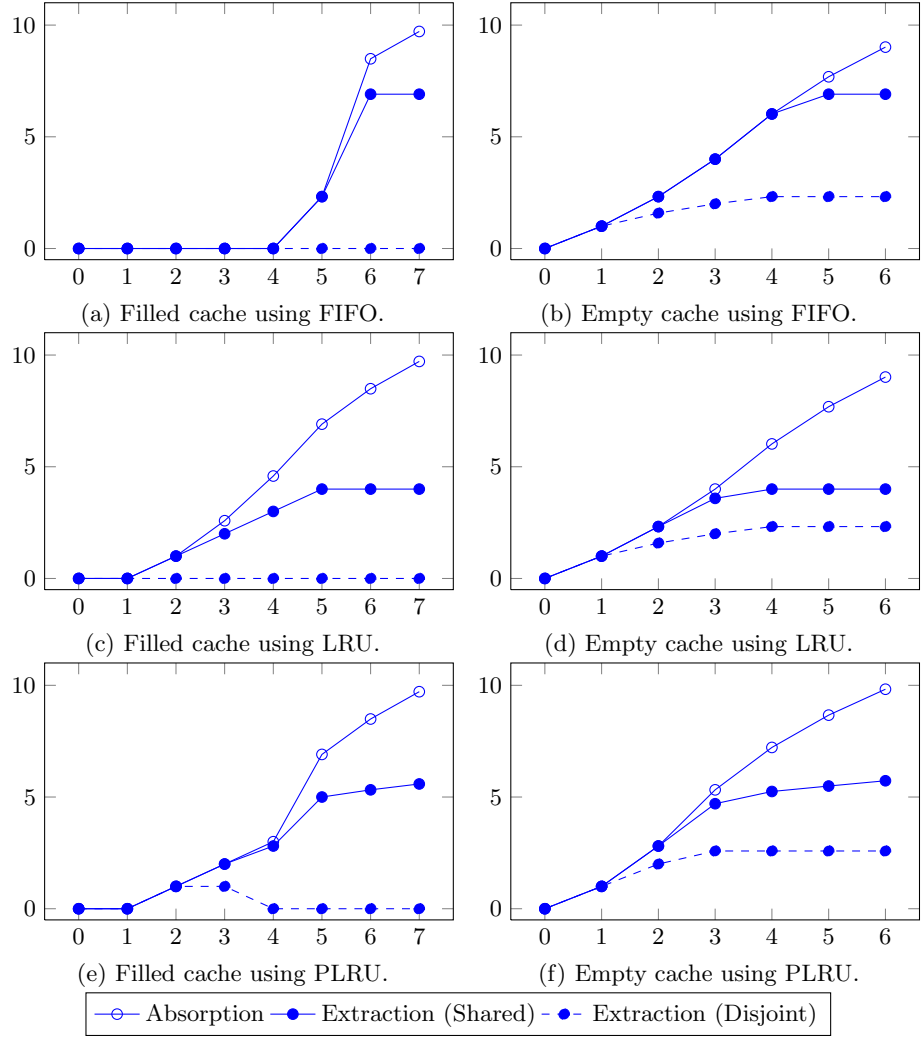
\begin{figure*}
\centering
\begin{tabular}{cc}
\begin{tikzpicture}
\begin{axis}[ymin=-0.5, ymax=11,
width=0.575\textwidth,height=5cm,
xtick={0,1,2,3,4,5,6,7},
legend columns=3,
legend entries={Absorption,Extraction (Shared),Extraction (Disjoint)},
legend to name=extr]
\AddCurve{Figures/abs_filled_fifo.txt}{blue}{o}{solid}
\AddCurve{Figures/ext_filled_fifo_s.txt}{blue}{*}{solid}
\AddCurve{Figures/ext_filled_fifo_d.txt}{blue}{*}{dashed}
\end{axis}
\end{tikzpicture}
\customlabel{fig:fifo_filled}{\ref*{fig:extraction}a}&
\begin{tikzpicture}
\begin{axis}[ymin=-0.5, ymax=11,
width=0.575\textwidth,height=5cm,
xtick={0,1,2,3,4,5,6,7}]
\AddCurve{Figures/abs_empty_fifo.txt}{blue}{o}{solid} 
\AddCurve{Figures/ext_empty_fifo_s.txt}{blue}{*}{solid}
\AddCurve{Figures/ext_empty_fifo_d.txt}{blue}{*}{dashed}
\end{axis}
\end{tikzpicture}
\customlabel{fig:fifo_empty}{\ref*{fig:extraction}b}\\

(a) Filled cache using FIFO.&(b) Empty cache using FIFO.
\end{tabular}

\begin{tabular}{cc}
\begin{tikzpicture}
\begin{axis}[ymin=-0.5, ymax=11,
width=0.575\textwidth,height=5cm,
xtick={0,1,2,3,4,5,6,7}]
\AddCurve{Figures/abs_filled_lru.txt}{blue}{o}{solid}
\AddCurve{Figures/ext_filled_lru_s.txt}{blue}{*}{solid}
\AddCurve{Figures/ext_filled_lru_d.txt}{blue}{*}{dashed}
\end{axis}
\end{tikzpicture}
\customlabel{fig:lru_filled}{\ref*{fig:extraction}c}&
\begin{tikzpicture}
\begin{axis}[ymin=-0.5, ymax=11,
width=0.575\textwidth,height=5cm,
xtick={0,1,2,3,4,5,6,7}]
\AddCurve{Figures/abs_empty_lru.txt}{blue}{o}{solid}
\AddCurve{Figures/ext_empty_lru_s.txt}{blue}{*}{solid}
\AddCurve{Figures/ext_empty_lru_d.txt}{blue}{*}{dashed}
\end{axis}
\end{tikzpicture}
\customlabel{fig:lru_empty}{\ref*{fig:extraction}d}\\

(c) Filled cache using LRU.&(d) Empty cache using LRU.
\end{tabular}

\begin{tabular}{cc}
\begin{tikzpicture}
\begin{axis}[ymin=-0.5, ymax=11,
width=0.575\textwidth,height=5cm,
xtick={0,1,2,3,4,5,6,7}]
\AddCurve{Figures/abs_filled_plru.txt}{blue}{o}{solid}
\AddCurve{Figures/ext_filled_plru_s.txt}{blue}{*}{solid}
\AddCurve{Figures/ext_filled_plru_d.txt}{blue}{*}{dashed}
\end{axis}
\end{tikzpicture}
\customlabel{fig:plru_filled}{\ref*{fig:extraction}e}&
\begin{tikzpicture}
\begin{axis}[ymin=-0.5, ymax=11,
width=0.575\textwidth,height=5cm,
xtick={0,1,2,3,4,5,6,7}]
\AddCurve{Figures/abs_empty_plru.txt}{blue}{o}{solid} 
\AddCurve{Figures/ext_empty_plru_s.txt}{blue}{*}{solid}
\AddCurve{Figures/ext_empty_plru_d.txt}{blue}{*}{dashed}
\end{axis}
\end{tikzpicture}
\customlabel{fig:plru_empty}{\ref*{fig:extraction}f}\\

(e) Filled cache using PLRU.&(f) Empty cache using PLRU.
\end{tabular}

\ref*{extr}
\caption{Information extraction of different replacement policies on a
  4-way cache set. Figures~\ref{fig:fifo_filled}, \ref{fig:lru_filled} and \ref{fig:plru_filled} depict the case of a
  filled initial cache, \ref{fig:fifo_empty}, \ref{fig:lru_empty} and \ref{fig:plru_empty} an empty one.  In
  all figures, the horizontal axis depicts the footprint, i.e., the
  number of memory blocks used. The vertical axis depicts the
  extracted information on a logarithmic scale, that is, in {\em
    bits}. The results for shared memory adversaries use the solid
  line; disjoint memory case uses the dashed line.
}
\label{fig:extraction}
\end{figure*}

%---------------------------------------------------------------------------------
%Results for AES 256
%---------------------------------------------------------------------------------
%

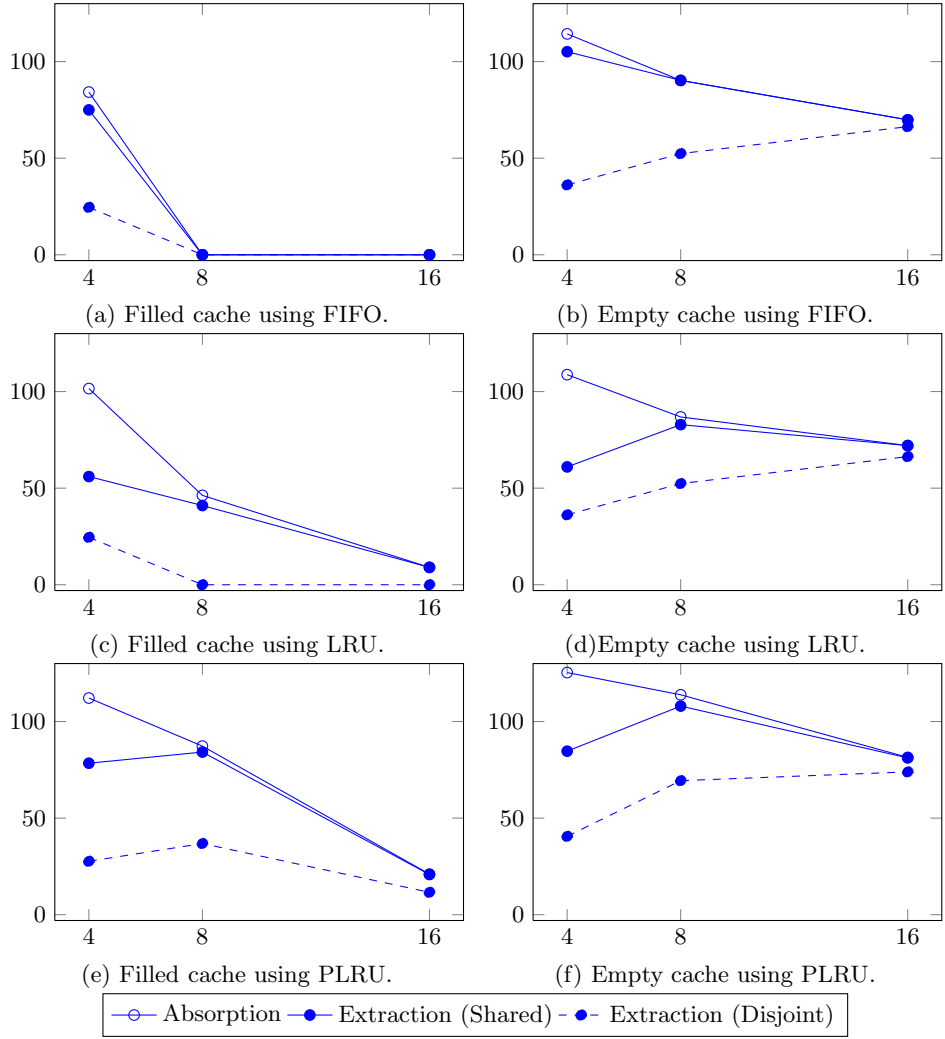
\begin{figure*}
\centering
\hspace*{-1em}
\begin{tabular}{cc}
\begin{tikzpicture}
\begin{axis}[ymin=-3,ymax=130,width=0.575\textwidth,height=5cm,xtick={4,8,16}]
\AddCurve{Cacheaudit/aes-256-preloading-fifo-abs-4.txt}{blue}{o}{solid}
\AddCurve{Cacheaudit/aes-256-preloading-fifo-extr-shared-4.txt}{blue}{*}{solid}
\AddCurve{Cacheaudit/aes-256-preloading-fifo-extr-dis-4.txt}{blue}{*}{dashed}
\end{axis}
\end{tikzpicture}
\customlabel{fig:fifo_aes_filled}{\ref*{fig:aes}a}&
\begin{tikzpicture}
\begin{axis}[ymin=-3,ymax=130,width=0.575\textwidth,height=5cm,xtick={4,8,16}]
\AddCurve{Cacheaudit/aes-256-rom-fifo-abs-4.txt}{blue}{o}{solid}
\AddCurve{Cacheaudit/aes-256-rom-fifo-extr-shared-4.txt}{blue}{*}{solid}
\AddCurve{Cacheaudit/aes-256-rom-fifo-extr-dis-4.txt}{blue}{*}{dashed}
\end{axis}
\end{tikzpicture}
\customlabel{fig:fifo_aes_empty}{\ref*{fig:aes}b}\\

(a) Filled cache using FIFO.&(b) Empty cache using FIFO.
\end{tabular}

\hspace*{-1em}
\begin{tabular}{cc}
\begin{tikzpicture}
\begin{axis}[ymin=-3,ymax=130,width=0.575\textwidth,height=5cm,xtick={4,8,16}]
\AddCurve{Cacheaudit/aes-256-preloading-lru-abs-4.txt}{blue}{o}{solid}
\AddCurve{Cacheaudit/aes-256-preloading-lru-extr-shared-4.txt}{blue}{*}{solid}
\AddCurve{Cacheaudit/aes-256-preloading-lru-extr-dis-4.txt}{blue}{*}{dashed}
\end{axis}
\end{tikzpicture}
\customlabel{fig:lru_aes_filled}{\ref*{fig:aes}c}&
\begin{tikzpicture}
\begin{axis}[ymin=-3,ymax=130,width=0.575\textwidth,height=5cm,xtick={4,8,16}]
\AddCurve{Cacheaudit/aes-256-rom-lru-abs-4.txt}{blue}{o}{solid}
\AddCurve{Cacheaudit/aes-256-rom-lru-extr-shared-4.txt}{blue}{*}{solid}
\AddCurve{Cacheaudit/aes-256-rom-lru-extr-dis-4.txt}{blue}{*}{dashed}
\end{axis}
\end{tikzpicture}
\customlabel{fig:lru_aes_empty}{\ref*{fig:aes}d}\\

(c) Filled cache using LRU.&(d)Empty cache using LRU.
\end{tabular}

\hspace*{-1em}
\begin{tabular}{cc}
\begin{tikzpicture}
\begin{axis}[ymin=-3,ymax=130,width=0.575\textwidth,height=5cm,xtick={4,8,16}]
\AddCurve{Cacheaudit/aes-256-preloading-plru-abs-4.txt}{blue}{o}{solid}
\AddCurve{Cacheaudit/aes-256-preloading-plru-extr-shared-4.txt}{blue}{*}{solid}
\AddCurve{Cacheaudit/aes-256-preloading-plru-extr-dis-4.txt}{blue}{*}{dashed}
\end{axis}
\end{tikzpicture}
\customlabel{fig:plru_aes_filled}{\ref*{fig:aes}e}&
\begin{tikzpicture}
\begin{axis}[ymin=-3,ymax=130,width=0.575\textwidth,height=5cm,xtick={4,8,16}]
\AddCurve{Cacheaudit/aes-256-rom-plru-abs-4.txt}{blue}{o}{solid}
\AddCurve{Cacheaudit/aes-256-rom-plru-extr-shared-4.txt}{blue}{*}{solid}
\AddCurve{Cacheaudit/aes-256-rom-plru-extr-dis-4.txt}{blue}{*}{dashed}
\end{axis}
\end{tikzpicture}
\customlabel{fig:plru_aes_empty}{\ref*{fig:aes}f}\\

(e) Filled cache using PLRU.&(f) Empty cache using PLRU.
\end{tabular}

\ref*{extr}
\caption{Information absorption and extraction (in bits) for the AES
  execution on a 4-way cache, for filled and empty initial cache states. Figures~\ref{fig:fifo_aes_filled}, \ref{fig:lru_aes_filled} and \ref{fig:plru_aes_filled} depict the case of a
  filled initial cache, \ref{fig:fifo_aes_empty}, \ref{fig:lru_aes_empty} and \ref{fig:plru_aes_empty} an empty one. The horizontal axis depicts the size of the cache in KB, the vertical axis depicts the extracted information in logarithmic scale.}
\label{fig:aes}
\end{figure*}

\section{Related Work}\label{sec:related}
Our work is related to existing models for adaptive probing~\cite{koepfbasin07,BorealeP12}. There, however, the secret remains static. The model of~\cite{koepfbasin07} and the deterministic part of~\cite{BorealeP12} is a special case of ours, where the update function is the identity. 

Mardziel et al.~\cite{mardziel14} develop an approach to quantify information flow for \emph{dynamic secrets}, that is, secrets that evolve over time. They consider a probabilistic system and attacks that consist of a fixed amount of steps. Attacks finish with an exploit whose success is evaluated using gain functions~\cite{alvim2012measuring}. Our model for information extraction differs from their model in that it is deterministic and allows to compute leakage for an undetermined number of attacks steps, i.e., until the probing is depleted. We further provide an algorithm that actually allows us to compute optimal strategies. We leave a probabilistic extension of our model to future work.

The problem that we consider in this paper is related to the {\em state identification} problem for Mealy machines, which was first introduced by Moore in \cite{Moore56}, expanded upon by Gill in \cite{Gill61}, and analyzed from a complexity perspective by Lee and Yannakakis~\cite{lee94}. The state-identification problem is to determine the initial state of a Mealy machine by probing strategies, just as in our case. While we are interested in the maximal number of knowledge sets into which the  uncertainty about the initial state can be partitioned, state-identification algorithms are only concerned with the decision problem, that is whether or not a full identification, i.e, a partitioning into singleton knowledge sets is feasible, and if it is, by which strategy.
So our problem of finding the finest partition can be seen as a quantitative generalization of the state-identification problem.

A proposal to quantify the security of cache memories was introduced in \cite{zhang14}. In this case, they use several types of attackers and study the security under different countermeasures, without considering the replacement policies individually. They obtained arguments in favor of some countermeasures against specific attacks. In our case we consider one single type of attacker, do not take into account any type of countermeasure and compare the different replacement policies.

\section{Future Work and Conclusions}

We presented a novel approach for quantifying isolation properties of shared caches, based on a simple model of adaptive attacks against Mealy machines.  We use our approach for performing the first security analysis of common cache replacement policies (LRU, FIFO, PLRU), as well as for improving the precision of the CacheAudit static analyzer.  Our prime target for future work is to investigate an extension of our model to Markov Decision Processes for dealing with randomized replacement policies.

\paragraph{Acknowledgments}
We thank Pierre Ganty and the anonymous reviewers for their
constructive feedback.

This work was supported by Microsoft Research through its PhD
Scholarship Programme, by Ram{\'o}n y Cajal grant RYC-2014-16766,
Spanish projects TIN2012-39391-C04-01 StrongSoft and
TIN2015-70713-R DEDETIS, and Madrid regional project S2013/ICE-2731
N-GREENS, and by the German
Research Council (DFG) as part of the Project PEP.

\bibliographystyle{abbrv}
\bibliography{leakmc}  

\begin{thebibliography}{10}

\bibitem{abel2013measurement}
A.~Abel and J.~Reineke.
\newblock Measurement-based modeling of the cache replacement policy.
\newblock In {\em RTAS}, pages 65--74. IEEE, 2013.

\bibitem{Aciicmez07Simp}
O.~Ac{\i}i{\c{c}}mez, {\c{C}}.~K. Ko{\c{c}}, and J.-P. Seifert.
\newblock On the power of simple branch prediction analysis.
\newblock In {\em ASIACCS}, pages 312--320. ACM, 2007.

\bibitem{Aciicmez07Diff}
O.~Ac{\i}i{\c{c}}mez, {\c{C}}.~K. Ko{\c{c}}, and J.-P. Seifert.
\newblock Predicting secret keys via branch prediction.
\newblock In {\em CT-RSA}, pages 225--242. Springer, 2007.

\bibitem{alvim2012measuring}
M.~S. Alvim, K.~Chatzikokolakis, C.~Palamidessi, and G.~Smith.
\newblock Measuring information leakage using generalized gain functions.
\newblock In {\em CSF}, pages 265--279. IEEE, 2012.

\bibitem{AskarovS07}
A.~Askarov and A.~Sabelfeld.
\newblock Gradual release: Unifying declassification, encryption and key
  release policies.
\newblock In {\em SSP}, pages 207--221. {IEEE}, 2007.

\bibitem{Bernstein05cache-timingattacks}
D.~Bernstein.
\newblock Cache-timing attacks on {AES}.
\newblock \url{http://cr.yp.to/antiforgery/cachetiming-20050414.pdf}, 2005.

\bibitem{BorealeP12}
M.~Boreale and F.~Pampaloni.
\newblock Quantitative multirun security under active adversaries.
\newblock In {\em QEST}. {IEEE}, 2012.

\bibitem{DenningWSM}
P.~J. Denning.
\newblock The working set model for program behavior.
\newblock {\em Commun. ACM}, 11(5):323--333, 1968.

\bibitem{doychev2015cacheaudit}
G.~Doychev, B.~K{\"o}pf, L.~Mauborgne, and J.~Reineke.
\newblock {CacheAudit}: a tool for the static analysis of cache side channels.
\newblock {\em ACM Transactions on Information and System Security (TISSEC)},
  18(1):4:1--4:32, 2015.

\bibitem{Gill61}
A.~Gill.
\newblock State-identification experiments in finite automata.
\newblock {\em Information and Control}, 4(2-3):132--154, 1961.

\bibitem{GullaschBK11}
D.~Gullasch, E.~Bangerter, and S.~Krenn.
\newblock Cache games - bringing access-based cache attacks on {AES} to
  practice.
\newblock In {\em SSP}, pages 490--505. IEEE, 2011.

\bibitem{koepfbasin07}
B.~K{\"o}pf and D.~Basin.
\newblock {An Information-Theoretic Model for Adaptive Side-Channel Attacks}.
\newblock In {\em CCS}, pages 286--296. ACM, 2007.

\bibitem{lampson1973note}
B.~W. Lampson.
\newblock A note on the confinement problem.
\newblock {\em Communications of the ACM}, 16(10):613--615, 1973.

\bibitem{lee94}
D.~Lee and M.~Yannakakis.
\newblock Testing finite-state machines: State identification and verification.
\newblock {\em {IEEE} Transactions on computers}, 43(3):306--320, 1994.

\bibitem{LiuYGHL15}
F.~Liu, Y.~Yarom, Q.~Ge, G.~Heiser, and R.~B. Lee.
\newblock Last-level cache side-channel attacks are practical.
\newblock In {\em SSP}, pages 605--622. {IEEE}, 2015.

\bibitem{mardziel14}
P.~Mardziel, M.~S. Alvim, M.~Hicks, and M.~R. Clarkson.
\newblock Quantifying information flow for dynamic secrets.
\newblock In {\em SSP}, pages 540--555. IEEE, 2014.

\bibitem{Moore56}
E.~F. Moore.
\newblock Gedanken-experiments on sequential machines.
\newblock {\em Automata studies}, 34:129--153, 1956.

\bibitem{NeveS06}
M.~Neve and J.~Seifert.
\newblock Advances on access-driven cache attacks on {AES}.
\newblock In {\em SAC}, pages 147--162. Springer, 2006.

\bibitem{osvikshamir06cache}
D.~A. Osvik, A.~Shamir, and E.~Tromer.
\newblock Cache attacks and countermeasures: the case of {AES}.
\newblock In {\em CT-RSA}, pages 1--20. Springer, 2006.

\bibitem{smith09}
G.~Smith.
\newblock On the foundations of quantitative information flow.
\newblock In {\em FoSSaCS}, pages 288--302. Springer, 2009.

\bibitem{TiwariOLVLHKCS11}
M.~Tiwari, J.~Oberg, X.~Li, J.~Valamehr, T.~E. Levin, B.~Hardekopf, R.~Kastner,
  F.~T. Chong, and T.~Sherwood.
\newblock Crafting a usable microkernel, processor, and {I/O} system with
  strict and provable information flow security.
\newblock In {\em ISCA}, pages 189--200. ACM, 2011.

\bibitem{xiang13hotl}
X.~Xiang, C.~Ding, H.~Luo, and B.~Bao.
\newblock {HOTL}: a higher order theory of locality.
\newblock In {\em ASPLOS}, pages 343--356. ACM, 2013.

\bibitem{YaromF14}
Y.~Yarom and K.~Falkner.
\newblock {FLUSH+RELOAD:} {A} high resolution, low noise, {L3} cache
  side-channel attack.
\newblock In {\em {USENIX}}, pages 719--732. USENIX Association, 2014.

\bibitem{zwsm15}
D.~Zhang, Y.~Wang, G.~E. Suh, and A.~C. Myers.
\newblock A hardware design language for timing-sensitive information-flow
  security.
\newblock In {\em ASPLOS}, pages 503--516. ACM, 2015.

\bibitem{zhang14}
T.~Zhang and R.~B. Lee.
\newblock New models of cache architectures characterizing information leakage
  from cache side channels.
\newblock In {\em ACSAC}, pages 96--105. ACM, 2014.

\end{thebibliography}
\end{document}